\documentclass[11pt]{article}

\usepackage{comment,amsmath,amssymb,geometry,mathtools,amsthm,amsfonts,xcolor,algorithm,algpseudocode,enumitem,todonotes,setspace,array}
\usepackage{tabularx}
\usepackage{booktabs}
\definecolor{darkblue}{rgb}{0,0,0.5}
\definecolor{darkred}{rgb}{0.5,0,0}
\usepackage[linktocpage=true,colorlinks,
urlcolor=darkblue, linkcolor=darkblue, anchorcolor=darkblue,citecolor=darkblue]{hyperref}

\usepackage[noabbrev,capitalise]{cleveref}

\numberwithin{equation}{section}
\newtheorem{theorem}{Theorem}[section]
\newtheorem*{theorem*}{Theorem}
\newtheorem{lemma}[theorem]{Lemma}
\newtheorem*{lemma*}{Lemma}
\newtheorem*{definition*}{Definition}
\newtheorem{proposition}[theorem]{Proposition}
\newtheorem*{proposition*}{Proposition}
\newtheorem{corollary}{Corollary}[section]
\newtheorem*{corollary*}{Corollary}
\newtheorem{definition}[theorem]{Definition}
\newtheorem*{definitions*}{Definitions}

\newtheorem{problem}{\bf Problem}
\newtheorem*{problem*}{\bf Problem}
\newtheorem{openproblem}{\bf Open problem}
\newtheorem*{openproblem*}{\bf Open problem}

\newtheorem*{example*}{\bf Example}
\theoremstyle{remark}
\newtheorem{remark}{\bf Remark}[section]

\DeclareMathOperator{\diag}{diag}

\DeclareMathOperator{\poly}{poly}
\DeclareMathOperator{\Orth}{\textbf{O}}
\DeclareMathOperator{\GL}{\textbf{GL}}

\DeclareMathOperator{\M}{\textbf{M}}
\DeclareMathOperator{\Orb}{Orb}
\DeclareMathOperator{\Unitary}{\textbf{U}}
\DeclareMathOperator{\supp}{supp}
\DeclareMathOperator{\vecc}{vec}

\newcommand{\R}{\mathbb{R}}
\newcommand{\F}{\mathbb{F}}
\newcommand{\C}{\mathbb{C}}

\newcommand{\E}{\mathbb{E}}
\newcommand{\Prob}{\mathbb{P}}
\newcommand{\Z}{\mathbb{Z}}

\newcommand{\LevyFunc}{\mathcal{L}}
\newcommand{\Comp}{\mathrm{Comp}}
\newcommand{\Incomp}{\mathrm{Incomp}}
\newcommand{\LCD}{\mathrm{LCD}}
\newcommand{\rLCD}{\widehat{\mathrm{LCD}}}
\newcommand{\dist}{\mathrm{dist}}
\newcommand{\spread}{\mathrm{spread}}

\newcommand{\event}{\mathcal{E}}
\newcommand{\net}{\mathcal{N}}
\newcommand{\MinSingVal}{s_{\mathrm{min}}}
\newcommand{\MaxSingVal}{s_{\mathrm{max}}}
\newcommand{\MinEigenvalueGap}{\lambda_{\mathrm{min}}}
\newcommand{\angles}[1]{\langle#1\rangle}

\renewcommand{\AA}{\mathbb{A}}

\newcommand{\CC}{\mathbb{C}}

\newcommand{\PP}{\mathbb{P}}
\newcommand{\QQ}{\mathbb{Q}}
\newcommand{\RR}{\mathbb{R}}

\newcommand{\ZZ}{\mathbb{Z}}

\newcommand{\ccPSPACE}{\textsf{\upshape PSPACE}}

\newcommand{\ccNP}{\textsf{\upshape NP}}
\newcommand{\ccAM}{\textsf{\upshape AM}}
\newcommand{\cccoAM}{\textsf{\upshape coAM}}

\newcommand{\ccGI}{\textsf{\upshape GI}}

\newcommand{\cK}{\mathcal{K}}

\newcommand{\cS}{\mathcal{S}}
\newcommand{\cT}{\mathcal{T}}

\newcommand{\Sym}{\mathrm{Sym}}

\newcommand{\OTI}{\textsc{O-TI}}
\newcommand{\UTI}{\textsc{U-TI}}
\newcommand{\GLTI}{\textsc{GL-TI}}

\title{Random tensor isomorphism under orthogonal and unitary actions}
\author{
Jeremy Chizewer\footnote{University of Chicago. \tt{jchizew@uchicago.edu}}
\and 
Samuel Everett\footnote{University of Chicago. \tt{same@uchicago.edu}}
\and 
Deven Mithal\footnote{University of Chicago. \tt{mithal@uchicago.edu}}
\and 
Youming Qiao\footnote{University of New South Wales. \tt{jimmyqiao86@gmail.com}}
}

\begin{document}
\pagenumbering{gobble}
\maketitle

\begin{abstract}
We study the problem of testing whether two tensors in $\mathbb{R}^\ell\otimes \mathbb{R}^m\otimes \mathbb{R}^n$ are isomorphic under the natural action of orthogonal groups $\textbf{O}(\ell, \mathbb{R})\times\textbf{O}(m, \mathbb{R})\times\textbf{O}(n, \mathbb{R})$, as well as the corresponding question over $\mathbb{C}$ and unitary groups. These problems naturally arise in several areas, including graph and tensor isomorphism (Grochow--Qiao, \emph{SIAM J. Comp.}'21), scaling algorithms for orbit closure intersections (Allen-Zhu--Garg--Li--Oliveira--Wigderson, \emph{STOC}'18), and quantum information (Liu--Li--Li--Qiao, \emph{Phys. Rev. Lett.}'12). 

We study average-case algorithms for orthogonal and unitary tensor isomorphism, with one random tensor where each entry is sampled uniformly independently from a sub-Gaussian distribution, and the other arbitrary. For the algorithm design, we develop algorithmic ideas from the higher-order singular value approach into polynomial-time exact (algebraic) and approximate (numerical) algorithms with rigorous average-case analyses. Following (Allen-Zhu--Garg--Li--Oliveira--Wigderson, \emph{STOC}'18), we present an algorithm for a gapped version of the orbit distance approximation problem.
For the average-case analysis, we work from recent progress in random matrix theory on eigenvalue repulsion of sub-Gaussian Wishart matrices (Christoffersen--Luh--O’Rourke--Shearer and Han, \emph{arXiv}'25) by extending their results from side lengths of Wishart matrices linearly related to polynomially related.

Our algorithmic and probabilistic results have natural implications and interpretations in several research directions. First, the tensor orbit closure intersection problems under orthogonal and unitary groups are average-case easy, in contrast to evidence for their worst-case hardness (B\"urgisser--Do\u{g}an--Makam--Walter--Wigderson, \emph{CCC}'21), as well as for the average-case tensor isomorphism under the general linear group action (Grochow--Qiao, \emph{SICOMP}'21). Second, one can test whether a Haar-random tripartite quantum state is locally unitary equivalent to an arbitrary tripartite quantum state in polynomial time. Third, we obtain a spectral algorithm for random tripartite 3-uniform hypergraph isomorphism.
\end{abstract}

\newpage

\tableofcontents

\newpage

\pagenumbering{arabic}
\setstretch{1.05}

\section{Introduction}

Given two mathematical objects, what is the computational cost of determining whether they are equivalent? These basic problems, known as \textit{isomorphism problems}, are of central importance in areas as varied as quantum information to graph theory. Substantial effort has been devoted to their study, but many fundamental questions remain unresolved.

One natural question arising in physics (starting with \cite{kraus2010local,LLLQ12}) is the matter of determining whether two tensors are isomorphic under orthogonal (over $\RR$) or unitary (over $\CC$) group actions. There are well-known algorithms for this problem that follow the classical higher-order singular value theory approach. However, no complexity analysis of the method has been obtained, and the average-case efficiency in particular is unknown. This situation could be compared to the understanding of the simplex algorithm before the work of Spielman and Teng \cite{spielman2004smoothed}.
The purpose of this paper is to develop the existing algorithmic ideas into average-case efficient algorithms.

In more formal terms, the \emph{unitary tensor isomorphism problem} ($\UTI$) asks whether two tensors $A, B \in \CC^{\ell} \otimes \CC^m \otimes \CC^n$ are equivalent under the natural action of the product of unitary groups $\Unitary(\ell, \CC)\times \Unitary(m, \CC)\times \Unitary(n, \CC)$. Replacing $\CC$ with $\RR$ and $\Unitary$ with $\Orth$, we obtain the \emph{orthogonal tensor isomorphism} problem ($\OTI$). Replacing $\Unitary$ with $\GL$, we obtain the \emph{general-linear tensor isomorphism} problem ($\GLTI$). In this paper, orthogonal groups are always over $\RR$ and unitary groups are always over $\CC$. For the purpose of exposition, we will mostly focus on $3$-tensors. 

\subsection{Motivations for orthogonal and unitary tensor isomorphism}\label{subsec:motivation}

\paragraph{Local unitary equivalence of quantum states.} Unitary tensor isomorphism appears naturally in quantum information. A tripartite pure quantum state is a unit tensor in $\C^\ell\otimes\C^m\otimes\C^n$. Two quantum states are locally unitary (LU) equivalent if and only if they are isomorphic under the unitary group action. Local unitary equivalence of quantum states is a basic equivalence relation of quantum states \cite{WGE16}, and captures the equivalence of two quantum states under Local Operations and Classical Communications (LOCC) \cite{dur2000three}. To decide whether two quantum states are LU equivalent has been studied in \cite{kraus2010local,LLLQ12}.

\paragraph{Cartesian tensors from data analysis.} The Singular Value Theorem states that the orbits of $\Orth(n, \RR)\times\Orth(m, \RR)$ on $\RR^n\otimes \RR^m$ are characterized by the singular values. In the study of data analysis, a central problem is to explore the generalization of the Singular Value Theorem to the statistical analysis of multiway data \cite{DL08}, and hence the orthogonal tensor isomorphism is of interest \cite{de2000multilinear,DL08,HU17,Sei18}. Tensors over $\RR$ with the orthogonal equivalence are sometimes referred to as Cartesian tensors \cite{CartesianTensor}.

\paragraph{Orbit closure intersection problems.} Orthogonal and unitary groups are compact, so the orbits of orthogonal and unitary group actions on tensors are closed (in Zariski or Euclidean topologies). Therefore, unitary and orthogonal tensor isomorphism are also instances of the family of orbit closure intersection problems. Given a group $G$ acting on a topological space $S$, the corresponding orbit closure intersection problem asks whether the orbit closures of $s, s'\in S$ intersect. 

With origins in the celebrated works of Hilbert and Mumford on invariant theory \cite{Hilbert,mumford1994geometric}, a complexity-theoretic study of orbit closure intersection problems received considerable attention with an algebraic complexity motivations \cite{HW15,GCT5}, while their significance was realized in diverse areas including combinatorics \cite{GGOW20,IQS18}, tensor networks \cite{acuaviva2023minimal}, matrix multiplication \cite{BI13}, statistics \cite{derksen2021maximum}, and counting problems \cite{cai2026vanishing}.

A major result in this research line is to settle the orbit closure intersection problem for the left-right action on matrix tuples \cite{GGOW20,IQS18,HH21,AGLOW18,IQ23}. After this, the next natural target is the orbit closure problems for tensors under the special linear group action \cite{burgisser2018efficient,AGLOW18}. The algorithmic technique in \cite{GGOW20,AGLOW18}, called operator scaling, leads to the development of the theory of non-commutative optimization \cite{BFGOWW19}. Briefly speaking, based on the Kempf--Ness theory, operator scaling reduces orbit closure problems for general or special linear group actions to those for the unitary or orthogonal group actions. As a result, unitary and orthogonal matrix tuple or tensor isomorphism arise naturally. Indeed, a key technical component of the algorithm for orbit closure intersection of the left-right action on matrix tuples in \cite{AGLOW18} is an algorithm for unitary matrix tuple equivalence.

\paragraph{Graph and tensor isomorphism problems.} Unitary and orthogonal tensor isomorphism are closely related to the celebrated graph isomorphism problem, and the general linear tensor isomorphism problem ($\GLTI$). $\GLTI$ received considerable attention due to its connections with group, polynomial, and algebra isomorphism \cite{TI1,TI2,TI3,TI4,TI5,JQSY19}.

Briefly speaking, unitary and orthogonal tensor isomorphism are ``sandwiched'' between graph isomorphism and GL tensor isomorphism. By \cite{TI3}, graph isomorphism poly-time reduces to unitary and orthogonal tensor isomorphism, and orthogonal tensor isomorphism poly-time reduces to GL tensor isomorphism. Answering an open question in \cite{TI3}, Lysikov and Walter showed that unitary tensor isomorphism reduces to GL tensor isomorphism in polynomial time \cite{lysikov2024complexity}. 

There has been extensive literature on graph isomorphism, with recent breakthroughs including Babai's quasipolynomial-time algorithm \cite{Bab16} and smoothed analysis of graph isomorphism by Anastos, Kwan, and Moore \cite{anastos2025smoothed}. There is also impressive progress on GL tensor isomorphism algorithms, thanks to Sun's breakthrough on $p$-group isomorphism \cite{Sun23,IMQSZ24}. By comparison, unitary and orthogonal tensor isomorphism, though closely related problems and interesting in their own right, remain relatively underexplored from an algorithmic perspective.

\subsection{Average-case complexities of isomorphism problems}\label{subsec:average-case}

Especially in their applied contexts, it is vital to understand the \emph{average-case complexity} of isomorphism problems. Recall that for an isomorphism problem, the input consists of two objects for which we wish to decide isomorphism. The average-case setting is to sample one object from some random model and to leave the other arbitrary. In particular, the second object may be chosen depending on the outcome of the random sample. The goal is to devise an algorithm testing isomorphism that works with high probability over the sampling of the first object. 

\paragraph{Average-case algorithms for graph isomorphism.}
Graph isomorphism is a classical example to demonstrate the difference between the worst-case and average-case complexities.
Despite the considerable attention the graph isomorphism problem (GI) has received, the worst-case complexity of the full GI remained exponential \cite{BL83} until Babai's breakthrough, a quasipolynomial-time algorithm for GI \cite{Bab16}. However, GI has long been regarded as easy to solve in practice \cite{mckay1981practical}, as evidenced by the software Nauty and Traces \cite{mckay2014practical}.

Such a contrast between worst-case and practical algorithms could be reconciled in theory by considering average-case complexity. Babai, Erd\H{o}s and Selkow first showed an average-case linear-time algorithm for graph isomorphism \cite{BES80} in the Erd\H{o}s--R\'enyi model $\mathcal{G}(n, p=1/2)$. This led to a series of works by e.g. Karp, Lipton, and Bollob\'as \cite{karp1975fast,lipton1978beacon,DBLP:conf/focs/BabaiK79,bollobas1982distinguishing,czajka2008improved,linial2017rigidity}, culminating in the recent work of Anastos, Kwan and Moore \cite{anastos2025smoothed}, who performed smoothed analysis on graph isomorphism and covered the canonical labelling of graphs sampled from the Erd\H{o}s--R\'enyi model $\mathcal{G}(n, p)$ for all $p$. 

\paragraph{Average-case algorithms for GL-tensor isomorphism over finite fields.} General-linear tensor isomorphism ($\GLTI$) over finite fields was proposed as a candidate in post-quantum cryptography in \cite{JQSY19}. Digital signatures based on $\GLTI$ or its variants were devised and implemented \cite{ALTEQ-paper,MEDS-paper}. 

Average-case algorithms play an important role in understanding the security of these cryptographic schemes. The first average-case algorithm for $\GLTI$ over finite fields was devised by Li and Qiao \cite{LQ17}, with improvement in \cite{BLQW20}. For $n\times n\times n$ tensors over $\F_q$, the average-case algorithm runs in time $q^{O(n)}$. In contrast, the best worst-case algorithm runs in time $q^{\tilde O(n^{1.5})}$ \cite{IMQSZ24} which is an improvement over Sun's breakthrough on $p$-group isomorphism \cite{Sun23}.

\paragraph{Spectral algorithms for random graph isomorphism.} There exist algorithms to test isomorphism of graphs with multiplicity free or bounded eigenvalues, by Leighton and Miller \cite{leighton1979certificates}\footnote{A nice exposition of \cite{leighton1979certificates} was recorded by Spielman \cite{spielman2018testing}.} and Babai, Grigoryev and Mount \cite{babai1982isomorphism,Bab86}. These algorithms motivated Babai to pose the question of whether the adjacency matrix of a random graph has a simple spectrum with high probability in the 1980s (cf. \cite{tao2017random}). This question was solved after three decades by Tao and Vu \cite{tao2017random}, achieved in the context of several breakthroughs in the inverse Littlewood--Offord theorems and universality in random matrix theory \cite{tao2009inverse,tao2011random}. As a result, the algorithms in \cite{leighton1979certificates,babai1982isomorphism,Bab86} can be interpreted as efficient algorithms for random graph isomorphism.

\subsection{Our results}\label{subsec:results}

Our main results are average-case efficient algorithms for the orthogonal and unitary tensor isomorphism problems. Before going into the details, we illustrate the position of our algorithms in the context of closely related isomorphism problems in Table~\ref{tab:comparison}.
\begin{table}[H]
\centering
\small
\setlength{\tabcolsep}{5pt}
\renewcommand{\arraystretch}{1.2}

\begin{tabularx}{\textwidth}{|
  >{\raggedright\arraybackslash}p{0.17\textwidth}|
  >{\raggedright\arraybackslash}p{0.27\textwidth}|
  >{\raggedright\arraybackslash}p{0.24\textwidth}|
  >{\raggedright\arraybackslash}p{0.23\textwidth}|}
\hline
 & \textbf{Orbit worst-case} & \textbf{Orbit average-case} & \textbf{Orbit closure intersection} \\
\hline
Graph Iso (GI) &
Quasipolynomial time \cite{Bab16} &
Polynomial time \cite{BES80,DBLP:conf/focs/BabaiK79,anastos2025smoothed} &
The same as orbit (in point-set topology) \\
\hline
Unitary ($\CC$) and Orthogonal ($\RR$) Tensor Iso &
$\ccGI$-hard \cite{TI3}; $\ccPSPACE$ &
\textbf{Polynomial time}\newline [This paper] &
The same as orbit (in Zariski or Euclidean topology) \\
\hline
General-linear Tensor Iso over $\C$ &
$\ccGI$-hard \cite{TI1};  $\ccAM$ \cite{koiran1996hilbert} or counting hierarchy \cite{AGS26} &
Open; related to orbit closure intersection &
$\ccGI$-hard \cite{lysikov2024complexity} \\
\hline
General-linear Tensor Iso over $\F_q$ &
$\ccGI$-hard \cite{TI1}; $\ccNP\cap\cccoAM$;
$q^{\tilde O(n^{1.5})}$-time algorithm \cite{Sun23,IMQSZ24}&
$q^{O(n)}$-time algorithm \cite{LQ17,BLQW20} &
Open (in Zariski topology) \\
\hline
\end{tabularx}
\caption{A comparison of several isomorphism problems in different perspectives (see Section~\ref{subsec:previous}).}
\label{tab:comparison}
\end{table}

\paragraph{Our random model of tensors.} A natural random model of tensors from $\CC^\ell\otimes\CC^m\otimes\CC^n$ is to sample each of the $\ell m n$ entries independently from a fixed---mean zero, unit variant---Gaussian distribution. This is in analogy with the classical Ginibre emsembles in the random matrix theory. Furthermore, this has a clear quantum information theoretic interpretation: after normalizing, this corresponds to a Haar-random pure state \cite{mezzadri2007generate}. 

Motivated by wider applications and universality results in random matrix theory (see Section~\ref{subsec:previous}), we actually adopt the random model of random tensors from $\RR^\ell\otimes\RR^m\otimes\RR^n$ where every entry is drawn independently from a fixed---mean zero, unit variance---sub-Gaussian distribution $\xi$. In the case of $\CC^\ell\otimes\CC^m\otimes\CC^n$, each entry is drawn independently from a fixed---mean zero, unit variance---sub-Gaussian distribution $\xi$ over the Gaussian rationals $\QQ(i) = \{s + it : s, t \in \QQ\}$. Recall that a real valued random variable $\xi$ is sub-Gaussian if it satisfies the Gaussian-type tail bound $\Prob[|\xi| > t] \leq 2\exp(-t^2/K^2)$ for some $K < \infty$. 

Sub-Gaussian distributions encompass many important distributions. Besides the Gaussian distribution, it also includes \emph{the Rademacher distribution}, i.e. the distribution on $\pm 1$ with equal probability. In this case, a random tensor in $\RR^\ell\otimes\RR^m\otimes\RR^n$ can be viewed as recording a random 3-uniform tripartite hypergraph $H=(L\cup M\cup N, E)$ where $|L|=\ell$, $|M|=m$, $|N|=n$, $E\subseteq L\times M\times N$, and each hyperedge $(i, j, k)$ is present in $H$ with probability $1/2$.

\paragraph{Our main results.}
To ease the statements of our results, we shall focus on tensors in $\CC^{n\times n\times n}=\CC^n\otimes\CC^n\otimes\CC^n$ (and similarly over $\RR$).
Formally, we begin by studying the following problem.

\begin{problem}\label{prob-1}
Provided tensors $A, B \in \CC^{n \times n \times n}$, are there $n \times n$ unitary matrices $L, R, T$ such that $(L, R, T) \curvearrowright A = B$?
\end{problem}

Our first result toward \cref{prob-1} is an average-case (in our random model) polynomial-time algorithm solving the \emph{exact} unitary and orthogonal TI problems:

\begin{theorem}[Global average-case analysis of the unitary tensor isomorphism problem]\label{thm:exactUnitIso}
There exists a deterministic algorithm which, given $A=(a_{i,j,k})\in\QQ(i)^{n\times n\times n}$ with i.i.d.\ entries $a_{i,j,k} \sim \xi$ for a fixed---mean zero, unit variance---sub-Gaussian distribution $\xi$, and $B\in\QQ(i)^{n\times n\times n}$ arbitrary and deterministic, runs in polynomial time on the input of $A$ and $B$, and with high probability over the choice of $A$
decides whether
\[
B\in\Orb_\CC(A)\coloneqq \{(L,R,T)\curvearrowright A:\ L,R,T\in \Unitary(n)\},
\]
with implied constants depending only on the distribution of $\xi$.
\end{theorem}

On the other hand, it is natural to consider the numerical, or approximation analogs of such problems. That is, we are also concerned with deciding the associated \emph{orbit distance approximation problems} (also referred to as robust orbit problems) as discussed in \cite{burgisser2024complexity} for instance.
Let
\[
\dist(X, Y)\coloneqq \inf_{x \in X, y \in Y} \|x-y\|_F
\]
denote the usual set distance using the Frobenius norm $\|\cdot\|_F$. 
Put
\[
\dist_\C(A,B) \coloneqq \dist(\Orb_\C(A), \Orb_\C(B)) = \inf_{L,R,T\in \Unitary(n)}\|(L,R,T)\curvearrowright A-B\|_F,
\]
where the equality holds since unitary actions are invertible, and the Frobenius norm is invariant under such actions.

We also study the following problem, following \cite[Theorem M3]{AGLOW18} which tackles a related problem (see Remark~\ref{remark:M3} for a detailed explanation).

\begin{problem}\label{prob-2}
The \emph{(gapped) unitary tensor orbit approximation problem} is defined as follows: Given $0<\varepsilon<\eta$ and tensors  $A,B \in \C^{n\times n \times n}$, distinguish between the following two cases: \[
\dist_\CC(A,B) \leq \varepsilon \; \text{ and }\; \dist_\CC(A,B) \geq \eta.
\]
\end{problem}

Of course, for a finite precision problem, we cannot actually draw the entries of $A$, nor can 
we even necessarily read in $B$ when considered true elements of $\C^{n\times n \times n}$. As such, we consider the model in which the algorithm is provided 
$\widetilde{A}, \widetilde{B}$, which are treated as truncated versions of $A$ and $B$ respectively, taken to the precision of our choice based on the provided $\varepsilon$.

The following theorem provides an average-case polynomial-time algorithm solving \cref{prob-2}.

\begin{theorem}[Global average-case analysis of the unitary orbit approximation problem]\label{thm:robustOrbit}
There exists a deterministic algorithm which, given $\varepsilon > 0$ such that $\varepsilon \le O(n^{-3})$, $A=(a_{i,j,k})\in \C^{n\times n\times n}$ with i.i.d.\ entries $a_{ijk} \sim \xi$ for a fixed---mean zero, unit variance---sub-Gaussian distribution $\xi$, and $B\in \C^{n\times n\times n}$ arbitrary, each specified to $O(\log(n/\varepsilon))$-bits of precisions, runs in polynomial-time in $n,\log(1/\varepsilon)$ and, with high probability over the choice of $A$, distinguishes 
$$\dist_\C(A,B)<\varepsilon\text{ from }\dist_\C(A,B)>\gamma\varepsilon\text{ for }\gamma = O(n^8).$$ In the former case, the algorithm also outputs a witness $(\widetilde{L},\widetilde{R},\widetilde{T})\in \Unitary(n)^3$ such that 
$$\|(\widetilde{L},\widetilde{R},\widetilde{T})\curvearrowright A-B\|_F< \gamma\varepsilon.$$
\end{theorem}

\begin{remark}\label{remark:M3}
It is instructive to compare Theorem~\ref{thm:robustOrbit} with \cite[Theorem M3]{AGLOW18}. 

In \cite{AGLOW18}, motivated by the orbit closure intersection problem of the left-right action on matrix tuples, the operator scaling technique reduces that problem to the orbit problem of the action of $\Unitary(n, \CC)\times\Unitary(n, \CC)$ on $\CC^n\otimes\CC^n\otimes\CC^m$. Note that there is no action on the $\CC^m$ component in the tensor product. For $A, B\in \CC^n\otimes\CC^n\otimes\CC^m$, we let $\widetilde\dist_\C(A,B)$ to denote the distance between the orbits of $A$ and $B$ under this action. \cite[Theorem M3]{AGLOW18} presented a polynomial-time algorithm to distinguish 
$$\widetilde \dist_\C(A,B)<\varepsilon\text{ from }\widetilde \dist_\C(A,B)>\gamma\varepsilon^{1/\poly(n, m)}\text{ for }\gamma = M^{\poly(n, m)},$$ 
where $M$ denotes the maximum magnitude over entries in $A$ and $B$. 

Compared with \cite[Theorem M3]{AGLOW18}, we note that Theorem~\ref{thm:robustOrbit} achieves better approximation ratio. On the other hand, note that the actions are different, and also our result is an average-case algorithm while \cite[Theorem M3]{AGLOW18} handles the worst-case setting.
\end{remark}

The average case analysis of \cref{thm:robustOrbit,thm:exactUnitIso} relies crucially on obtaining a subtle understanding of the spectra of random matrices. Specifically, we rely on the following repulsion result, proved in \cref{sec-repulsion}.

\begin{theorem}[Eigenvalue repulsion with polynomial gap]\label{thm-repulsion-unscaled}
    Fix a mean zero, variance one sub-Gaussian distribution $\xi$. Let $M$ be a $n \times p$ matrix with i.i.d.\ entries $m_{ij} \sim \xi$, and with $p = \Theta(n^{\zeta})$ for $\zeta \in (0, 1)$. Then, defining
    \begin{equation*}
        \MinEigenvalueGap(M^\top  M) \coloneqq \inf_{1 \leq i < p} (\lambda_i(M^\top  M) - \lambda_{i + 1}(M^\top  M))
    \end{equation*}
    and with $\beta > \zeta$, we have that
    \begin{equation*}
        \MinEigenvalueGap(M^\top  M) \geq \Omega((n^{(1 - \zeta)/2} - 1)n^{- \beta})
    \end{equation*}
    with probability at least $1 - O(n^{\zeta - \beta})$. Additionally, letting $M$ have complex entries $m_{ij} = a_{ij} + ib_{ij}$ where $a_{ij}, b_{ij}$ are i.i.d.\ real variables $\sim \xi$ then an identical bound holds for the eigenvalues of $M^\top M$.
\end{theorem}

This result is easily converted from the $n \times n^{1/k}$ regime to $n^k \times n$ with $k \in \R_+$ by substitution, yielding a gap
\begin{equation}\label{eq-gap-useable}
    \MinEigenvalueGap(M^\top  M) \geq \Omega((n^{(k - 1)/2} - 1)n^{- k\beta})
\end{equation}
with probability at least $1 - O(n^{1 - k\beta})$ where $\beta > 1/k$, where $k$ denotes the order of the tensors.

\begin{remark}[Probability versus precision tradeoff]
    In \cref{thm-repulsion-unscaled} the target gap appears in both the repulsion event and upper bound on the probability. This is a common feature of inverse Littlewood--Offord based repulsion results (see \cite{nguyen2015randommatricestailbounds, christoffersen2025gaps, han2025simplicity}) and in our setting enables a clean tradeoff between success probability and the precision demands of our numerical algorithms.
\end{remark}

A consequence of the analysis of \cref{thm:exactUnitIso,thm:robustOrbit,thm-repulsion-unscaled}, as well as the spectral stability corollary of the Weyl inequality, is that we are able to make the following smoothed analysis/local average-case complexity statements. Their proofs are trivial, so we omit them.

\begin{theorem}[Local analysis of the unitary tensor isomorphism problem]\label{thm:exactUnitIsoLocal}
There exists a deterministic algorithm which, given $E=(e_{ijk})\in\QQ(i)^{n\times n\times n}$ with i.i.d.\ entries $e_{ijk} \sim \xi$ for a fixed---mean zero, unit variance---sub-Gaussian distribution $\xi$, a constant $\beta \in (1/3, 1/2)$, a perturbation magnitude $\eta \geq 2 \cdot \|A\|(\|A\| + \sqrt{n} + K\sqrt{p})\cdot n^{3\beta}/(n + 1)$ and $A, B\in\QQ(i)^{n\times n\times n}$ arbitrary and deterministic, runs in polynomial-time on the input of $A + \eta E$ and $B$ and decides with high probability whether
\[
B\in\Orb(A)\coloneqq \{(L,R,T)\curvearrowright A:\ L,R,T\in \Unitary(n)\},
\]
with implied constants depending only on the distribution of $\xi$.
\end{theorem}

\begin{theorem}[Local analysis of the unitary orbit approximation problem]\label{thm:robustOrbitLocal}
There exists a deterministic algorithm which, given $\varepsilon > 0$ such that $\varepsilon =O(n^{-3})$, $E=(e_{ijk})\in\R^{n\times n\times n}$ with i.i.d.\ entries $e_{ijk} \sim \xi$ for a fixed---mean zero, unit variance---sub-Gaussian distribution $\xi$, a constant $\beta \in (1/3, 1/2)$, a perturbation magnitude $\eta \geq 2 \cdot \|A\|(\|A\| + \sqrt{n} + K\sqrt{p})\cdot n^{3\beta}/(n + 1)$ and $A, B \in \C^{n\times n\times n}$ arbitrary and deterministic, each specified to $O(\log(n/\varepsilon))$-bits of precisions, runs in polynomial-time in $n,\log(1/\varepsilon)$ and distinguishes $\dist_\C(A + \eta E,B)<\varepsilon$ from $\dist_\C(A + \eta E,B)>\gamma\varepsilon$ for $\gamma = O(n^8)$ with high probability, and in the former case outputs a witness $(\widetilde{L},\widetilde{R},\widetilde{T})\in \Unitary(n)^3$ such that
\[
\|(\widetilde{L},\widetilde{R},\widetilde{T})\curvearrowright (A + \eta E)-B\|_F< \gamma\varepsilon.
\]
\end{theorem}
The perturbation size $\eta$ is chosen such that through spectral stability one can deduce the preservation of the simplicity of the spectrum of $(A + \eta E)^T(A + \eta E)$ under the additive perturbation from $A$. These simple theorems are intended to serve as a baseline for more sophisticated smooth analysis (see \cref{sec-open-problems}).

Indeed, notice that for $k$ sufficiently large (with $k$ again denoting the order of the tensors), and a sufficiently small target success probability, by \cref{thm-repulsion-unscaled} the spectral gap can grow, in which case the perturbation size $\eta$ in \cref{thm:exactUnitIsoLocal,thm:robustOrbitLocal} can be taken to shrink. On the other hand, when $k$ is small, the perturbation $\eta$ is quite large. This indicates a qualitative distinction between the smoothed analysis setups for $k$-TI as $k$ varies, and indicates a limitation of our smoothed analysis.

\paragraph{Applications of our results.} Our main results have the following interpretations. 

First, we instantiate the sub-Gaussian distribution in Theorem~\ref{thm:exactUnitIso} with the Gaussian distribution, and view unit tensors in $\CC^n\otimes \CC^n\otimes \CC^n$ as pure quantum states. Suppose we are given a quantum state as an $n\times n\times n$ tensor. Then we have: 
\begin{corollary}
    There is a polynomial-time algorithm that decides whether a Haar-random quantum state and an arbitrary state in $\CC^n\otimes \CC^n\otimes \CC^n$ are locally unitary equivalent. 
\end{corollary}

Second, we instantiate the sub-Gaussian distribution with the Rademacher distribution, and view $\pm 1$ tensors in $\RR^n\otimes \RR^n\otimes \RR^n$ as adjacency tensors of 3-uniform tripartite hypergraphs. Combining our Theorem~\ref{thm-repulsion-unscaled} with the Leighton and Miller's spectral algorithm for graph isomorphism \cite{leighton1979certificates} (see \cite{spielman2018testing}), we have the following result, whose proof sketch is in Appendix~\ref{app:hypergraph-iso}.  

\begin{corollary}\label{cor:hypergraph-iso}
    There is a polynomial-time algorithm that decides whether a random 3-uniform tripartite hypergraph and an arbitrary 3-uniform tripartite hypergraph are isomorphic.
\end{corollary}

Note that Corollary~\ref{cor:hypergraph-iso} can be viewed as a hypergraph analogue of the spectral algorithms for random graph isomorphism discussed at the end of Section~\ref{subsec:average-case}. It cannot be deduced from these algorithms because the reduction from hypergraph isomorphism to graph isomorphism \cite{KoblerEtAl1993} does not preserve the distributions underlying the average-case analyses.

\subsection{Overview of techniques}\label{subsec:techniques}

Here we illustrate the principal tools and techniques used in proving the main results of this paper. 

\subsubsection{Algorithmic techniques}

\paragraph{An Isomorphism Invariant.} Let $A=(a_{i,j,k})\in \CC^{n\times n\times n}$. A classical isomorphism invariant of $A$ under action of $\Unitary(n)^3$ is the following. First, slicing along the third index, we obtain a matrix tuple $(A_1, \dots, A_n)$, where $A_k(i, j)=a_{i,j,k}$; this is treated as an $n \times n^2$ matrix $\hat{A}$. Construct the Gram matrix $G_3(A)=\hat{A}\hat{A}^*$ where $\hat{A}^*$ is the conjugate transpose of $\hat{A}$, and $3$ in the subscript indicates that the Gram matrix comes from slicing along the third direction. It is clear that $G_3(A)$ is invariant under the action of the first two components of $\Unitary(n)^3$ (i.e.\ in $(L, R, T) \curvearrowright A$). The third component $T\in \Unitary(n)$ sends $G_3(A)$ to $TG_3(A)T^*$. As a result, the spectrum of $G_3(A)$ is an isomorphism invariant of $\Unitary(n)^3$ on $A=(a_{i,j,k})\in \CC^{n\times n\times n}$.

If the spectrum of $G_3(A)$ is simple, that is every eigenvalue appears with multiplicity $1$, then these $1$-dimensional eigenspaces need to be matched accordingly in an isomorphism from $A$ to an isomorphic $B$. Note that eigenvalues of $G_3(A)$ correspond to singular values of the matrix $\hat{A}$. Suppose that this also holds for $G_1(A)$ and $G_2(A)$, the Gram matrices obtained by slicing along two other directions. Then it is not hard to see that testing isomorphism between $A$ and another $B\in \CC^{n\times n\times n}$ can be done similarly; that the spectrum align of the Gram matrices $G_i(A)$ and $G_i(B)$ is a necessary albeit not sufficient condition for $A$ and $B$ to be unitarily isomorphic.

\paragraph{Exact Unitary Tensor Isomorphism.}
We now outline the algorithm used to establish \cref{thm:exactUnitIso}. First, compute the eigenvalues and their eigenspaces for $G_{d}(A)$, $d=1, 2, 3$, defined above. For $A$ and $B$ to be isomorphic, $G_{d}(B)$ needs to have the same eigenvalues with the same multiplicities as $G_{d}(A)$ for every $d=1, 2, 3$. Then, the problem reduces to that of deciding whether $A$ and $B$ are isomorphic under invertible diagonal matrices. Finally, we solve the corresponding tensor diagonal isomorphism problem.

To implement the above algorithm outline requires some care. For example, while computing eigenvalues and eigenspaces is a basic task in computational linear algebra, rigorous analysis with provable guarantees can be tricky. An exact algorithm in the algebraic model was obtained by Cai \cite{cai1994computing}, and a numerical algorithm with a rigorous analysis was only recently shown by Dey, Kannan, Ryder and Srivastava \cite{dey2023bit}. Depending on the computing model (algebraic or numerical), we can make use of \cite{cai1994computing} or \cite{dey2023bit}. As another example, to solve tensor diagonal isomorphism, we resort to the recent work on torus actions \cite{burgisser_et_alCCC.2021.32}. These are put together to achieve Theorem~\ref{thm:exactUnitIso}.

\paragraph{Unitary Orbit Approximation.}
A central piece of the proof is the application of the Higher-Order Singular Value Decomposition (HOSVD), as in the exact unitary isomorphism case. However, we must develop additional techniques to modify the algorithm for the ``robust'' setting. The existing HOSVD algorithm determines whether two tensors are exactly isomorphic. Answering the robust version of the problem requires non-trivial modifications to the algorithm.

In essence, the HOSVD algorithm computes the mode-$d$ flattening of each tensor for each mode $d=1,2,3$, and then computes the mode-$d$ Gram matrix by multiplying each flattened tensor by its conjugate transpose. Then, the algorithm computes the spectra of these Gram matrices and compares them. In the exact version of the problem, we also require that these spectra are simple, and the two tensors are isomorphic only if the simple spectra match for corresponding Gram matrices of the tensors. If the spectra do match, we establish a system of linear equalities which has a solution corresponding to a set of ``phase-matrices'' if and only if the tensors are isomorphic. 

In our robust version of the algorithm, we again compute the Gram matrices, and this time require that the spectra have a sufficiently large eigenvalue gap (analogous to the simple spectra condition in the exact setting). Now, instead of checking whether or not the spectra are exactly equal, we check whether they are \emph{sufficiently close}. Then, instead of establishing a system of equalities, we establish a system of \emph{inequalities} corresponding to a set of approximate phase-matrices. While these modifications are natural, proving correctness of the new algorithm is the main hurdle. 

To prove correctness we must show that our approximate solution is not too far from the true orbit distance. We first use Weyl's inequality to relate the distance bound on the tensors to a distance bound on the Gram matrices. Then we use a probability trick, interpreting the diagonal entries of a Gram-matrix as the values of a random variable and a row of a unitary matrix as the corresponding probabilities of those values. In this view, we can bound the distance from the optimal solution using the variance of this imagined random variable. Note that we do not rely on the results of \cite{burgisser_et_alCCC.2021.32} for the unitary orbit approximation.

\subsubsection{Probabilistic techniques}
\paragraph{Average-case analysis.} The key to our algorithm is to assume that $G_{d}(A)$ has a simple spectrum with some quantitative repulsion. The proof strategy we adopt, initially established in \cite{nguyen2015randommatricestailbounds} in the context of Hermitian random matrices and then adapted to sub-Gaussian Wishart matrices in \cite{christoffersen2025gaps}, follows largely the structure of \cite{christoffersen2025gaps}. Both papers tackle the minimal eigenvalue gap problem by reducing it to producing a lower bound on a certain small ball probability associated with the random matrix $M$. In \cite{nguyen2015randommatricestailbounds} this amounts to controlling the concentration properties of the inner product of an eigenvalue $v$ of a submatrix $M'$ of $M$ with an independent sub-Gaussian vector, but in \cite{christoffersen2025gaps} the dependence structure of $M = A^\top A$ with $A$ populated by i.i.d.\ entries is more complex and so the authors control properties of both the left and right singular vectors associated with $M$. In both cases, the tools used to control the inner product are from an area known as ``inverse Littlewood--Offord theory''. The primary technical hurdle our regime ($p = o(n)$) is that much of the symmetry of the argument in \cite{christoffersen2025gaps} when dealing with left and right singular vectors breaks. We discuss the results of \cite{christoffersen2025gaps}, and the way in which we build on them, in \cref{sec-clos-proof-structure}.

There are two ways in which the quantity $|v'X|$ can be small: 1) the singular vector has relatively many zeroes (or entries close to zero), which is captured by the notion of \textit{compressibility}, or 2) the singular vector has substantial additive structure, and therefore a possibility of cancellation in the inner product, which is captured by the notion of \textit{low LCD}, where LCD here stands for Least Common Denominator. The aforementioned papers work by ruling out conditions 1) and 2) successively. We rule out both properties for the left singular vectors, but for the right singular vectors we need only establish 1).

\paragraph{Heuristics for the repulsion.}
We briefly discuss a pair of heuristics for why we should expect the particular quantitative repulsion seen in \cref{thm-repulsion-unscaled}.
In particular, one may notice that in \cref{thm-repulsion-unscaled} that if $p = n^{\zeta}$ for $\zeta > 0$ sufficiently small, that there exist choices of the target gap (configured by sufficiently small $\beta$) such that the extremal eigenvalue gap of $M^\top M$ actually \textit{grows}. At first glance this may be somewhat surprising, so we provide a brief explanation.

First, for an $n \times p$ sub-Gaussian matrix $M$ with $n > p$, we have that with overwhelming probability the singular values of $M$ sit in a window $[\sqrt{n} - \Theta(\sqrt{p}), \sqrt{n} + \Theta(\sqrt{p})]$ (see \cref{subsec-prob-prelims}). As there are $p$ singular values, a \textit{single} singular value gap is roughly on the order of $1/\sqrt{p}$. As the extremal singular gap is smaller, we take as our ansatz a singular value gap of the form $1/p^{\Delta}$ with $\Delta > 1/2$ and independent of $\zeta$. Now, the $i$'th gap between \textit{square} singular values can be written as the difference of squares $\sigma_i^2(M) - \sigma_{i + 1}^2(M) = (\sigma_i(M) - \sigma_{i + 1}(M))(\sigma_i(M) + \sigma_{i + 1}(M))$ which by our ansatz and overwhelming bound on the smallest singular value gives $\sigma_i^2(M) - \sigma_{i + 1}^2(M) \geq 2(\sqrt{n} - \Theta(\sqrt{p}))/(p^{\Delta}) = 2(n^{1/2} - \Theta(n^{\zeta/2}))/(n^{\zeta \cdot \Delta})$ for all $1 \leq i \leq p$. Clearly, when $\zeta$ is sufficiently small, the $\sqrt{n}$ term in the numerator dominates, and the extremal gap diverges as $n \to \infty$. Note though that to achieve the strongest repulsion probability in \cref{thm-repulsion-unscaled} that one takes $\beta$ arbitrarily large, and so naturally winds up studying a shrinking gap.

Secondly, we describe a limited case in which our results for the structure of the right singular vectors (\cref{lem-incompressible-fixed-val-right}) tall case can be deduced from the linearly related results of \cite{han2025simplicity, christoffersen2025gaps}. If the matrix $M$ has i.i.d.\ entries $\sim \mathcal{N}(0, 1)$, i.e. the classical Wishart ensemble (denoted $W(p, n)$), then the matrix ensemble $W(p, n)$ is orthogonally invariant and therefore has eigenvectors which are Haar distributed on $\Orth(p)$ and independent of the eigenvalues regardless of $n$ (see \cite{anderson2010introduction}), and the structural conclusions of the eigenvectors of the ensembles $W(p, p)$ and $W(p, n)$ with $n > p$ are identical. This also motivates the assumption that $\Delta$ and $\zeta$ are independent in the spacing heuristic.

\subsection{Open problems}\label{sec-open-problems}
Our main results, as well as the techniques we use to prove them, have led us to identify the following salient open problems. These problems are those most directly impacting our ability to obtain a more complete understanding of TI problems, especially as it pertains to their average-case complexity.

\begin{openproblem}[General linear group tensor isomorphism.]
Perhaps the principal problem our work leaves open to further study is the question of whether 3-tensor isomorphism under general linear group actions is easy on average. While $3$-tensor isomorphism remains exponentially hard over finite fields \cite{LQ17,BLQW20}, it may be the case that more tools and techniques could be developed for $\RR$ or $\CC$.
\end{openproblem}

\begin{openproblem}[Smoothed analysis for Tensor Orthogonal and Unitary Isomorphism.]
    In the study of Higher-Order Singular Value Decomposition, it is usually assumed that the marginal spectra are simple. This prompts a smoothed analysis \cite{spielman2004smoothed} study of this phenomena, namely a small perturbation of the entries of a fixed tensor would result in marginal spectra being simple. Our Theorem~\ref{thm:exactUnitIsoLocal} is motivated by such a consideration  and a first step in this direction.
\end{openproblem}

\begin{openproblem}[Average-case algorithms for Tensor Orbit Closure Intersection.] As described in Section~\ref{subsec:motivation}, Tensor Orbit Closure Intersection (TOCI) has been a major problem in non-commutative optimization, but some works provide certain barriers for existing techniques. However, to the best of our knowledge, there are no formal average-case algorithmic studies for such problems yet. Given that Tensor Unitary Isomorphism plays a key role in several algorithms there, it is expect that our average-case algorithms for Tensor Unitary Isomorphism will motivate research into this direction.
\end{openproblem}

\begin{openproblem}[Tensor Orthogonal and Unitary Isomorphism over finite fields.]
Tensor Orthogonal Isomorphism can also be studied over finite fields. While the spectra for random matrices have been studied \cite{NEUMANN_PRAEGER_1998,luh2021some}, singular values of Wishart matrices over finite fields do not appear to be studied yet, and it is unclear if there would be singular values of the original finite field.
\end{openproblem}

\begin{openproblem}[From simple spectra to bounded multiplicity.]
Another intriguing problem is to relax from simple spectra (each eigenvalue appearing with multiplicity $1$) to bounded multiplicity (each eigenvalue appearing with multiplicity at most $c$ for a constant $c$). In this setting, efficient algorithms based on invariant theory (following \cite{burgisser_et_alCCC.2021.32}) may be feasible, as the ``non-commutative parts'' are rather restricted so exponential upper bounds on the degree of generation \cite{derksen2001polynomial} could be used here.
\end{openproblem}

\subsection{A broader perspective on previous work}\label{subsec:previous}

\paragraph{$3$-tensors and $k$-tensors.} Tensor isomorphism problems naturally generalizes from order-$3$ tensors to tensors of arbitrary order. Interestingly, for any $k>3$, the general linear $k$-tensor isomorphism reduces to the general linear $3$-tensor isomorphism \cite{TI1} in polynomial time, as well as the corresponding unitary and orthogonal versions \cite{TI3}. For this reason we will restrict to the consideration of 3-tensor isomorphism problems, and abuse notation by letting TI denote both tensor isomorphism problems, as well as the complexity class (we will similarly abuse notation for the class GI).

\paragraph{Average-case algorithms for graph isomorphism.} After the first average-case algorithm for graph isomorphism by Babai, Erd\H{o}s, and Selkow \cite{BES80}, the failure probability of average-case algorithms was improved in several works by Karp \cite{karp1975fast}, Lipton \cite{lipton1978beacon}, and Babai and Ku\v{c}era \cite{DBLP:conf/focs/BabaiK79}. Random graph isomorphism in the general Erd\H{o}s--R\'enyi model $\mathcal{G}(n, p)$, $p$ not necessarily $1/2$, has been studied by Bollob\'as \cite{bollobas1982distinguishing}, Czajka and Pandurangan \cite{czajka2008improved}, and Linial and Mosheiff \cite{linial2017rigidity}. Random hypergraph isomorphism was discussed in \cite{chakraborti2021isomorphism}.

\paragraph{Random matrix theory.}
In random matrix theory, a productive line of modern research considers the understanding of global and local spectral statistics of a given random matrix ensemble. Particular quantities of interest are the invertibility property of a random matrix, the asymptotic behavior of the gap between adjacent eigenvalues, and the behavior of the smallest of all such gaps. As one varies the properties of the random matrix ensemble, such as the marginal distribution of an entry, the dependence between entries, and the dimensions of the matrix, these spectral properties can vary considerably.

The question of the invertibility likelihood of random matrices saw major progress with the line of work \cite{rudelson2008littlewood, tao2009inverse, tao2009littlewood, rudelson2009smallest, vershynin2014invertibility} addressing comprehensively the task for sub-Gaussian random matrices, and furnishing the random matrix theory literature with tools such as Littlewood--Offord theory and its inverse analogue.

Roughly, given a deterministic vector $x \in \R^n$ and a random vector $X \in \R^n$, the Littlewood--Offord problem asks about the impact of the structure of $x$ is on the probability of the event that $\langle x, X \rangle$ is large. Conversely, the inverse Littlewood--Offord problem asks about the structural implications for $x$ if $\langle x, X \rangle$ is only large with small probability \cite{tao2009inverse}. As is reflected in \cite{nguyen2015randommatricestailbounds, han2025simplicity, christoffersen2025gaps}, and indeed in our own work, the insights of inverse Littlewood--Offord theory provide a clean enumeration of properties of eigenvectors which one must rule out to deduce eigenvalue repulsion.

Studying the gaps between eigenvalues is typically easier in the ``bulk'' of the spectrum (i.e. away from the biggest and smallest eigenvalues) and results on local eigenvalue statistics for the bulk of Wigner and Wishart matrices were obtained in \cite{tao2011random} and \cite{Tao_2012} respectively, before being extended to the edge of the spectrum in \cite{Tao_2010} and \cite{wang2012random} respectively. Notably, these results are universal in that they apply generically to matrix ensembles with coordinate distributions satisfying mild regularity properties, typically pertaining to concentration of measure.

The control of the extremal gap has been naturally more elusive, but initially under strong distributional assumptions progress was made in papers such as \cite{ben2013extreme} which deals with the GUE matrix ensemble. While the unitarily and orthogonally Gaussian ensembles admitted significant analytic tractability, the qualitative properties of the extremal gap were expected to be universal, with control of the tail bounds of the extremal gap and consequently the simplicity likelihood of the spectrum in a universal setting being achieved in \cite{nguyen2015randommatricestailbounds} for the Wigner case. The prior result used the aforementioned inverse Littlewood--Offord theory, and very recently through inverse Littlewood--Offord techniques similar properties were proven concurrently in a universal setting by \cite{christoffersen2025gaps} and \cite{han2025simplicity} for Wishart matrices, achieving parity with the Wigner case.

It should be noted that while techniques using inverse Littlewood--Offord theory, and approaches from combinatorial random matrix theory overall \cite{Vu14, Vu21}, have provided widely-used and effective tools for studying local eigenvalue statistics, new ideas and tools are also required to make progress on extremal gaps, as seen in papers such as \cite{erdHos2010wegner, erdHos2015gap}.

\paragraph{Remarks on some entries in Table~\ref{tab:comparison}.} Some entries in Table~\ref{tab:comparison} deserve additional explanation.
\begin{itemize}
    \item Unitary and orthogonal tensor isomorphism is in $\ccPSPACE$ because we can formulate these problems as solvability of polynomial equations over $\RR$. This is achieved by first setting variable matrices for the transformation matrices; in the case of $\Unitary(n, \C)$, we need to record the real and imaginary parts separately. Then we can express two tensors being isomorphic as a set of polynomial equations over $\RR$. Solving polynomial equations over $\RR$ is known to be in $\ccPSPACE$ by \cite{Canny88}.
    \item General-linear tensor isomorphism over $\CC$ can also be formulated as solving a system of polynomial equations. Here we don't need to record the real and imaginary parts separately (as variables over reals) because the conjugate transpose is not required. The results in \cite{koiran1996hilbert} and \cite{AGS26} represent the best-known complexity-theoretic results for Hilbert's Nullstellensatz. Note that \cite{koiran1996hilbert} assumes Generalized Riemann Hypothesis and the tensors have integer entries. 
    \item Average-case general-linear tensor isomorphism over $\CC$ is related to its orbit closure intersection problem, because tensors whose orbits are closed form a Zariski open set. Therefore, for natural models of random tensors, a random tensor would have a closed orbit, and for such tensors, orbit and orbit closure intersection problems are the same.
    \item General-linear tensor isomorphism over $\F_q$ is in $\cccoAM$ by following the same interactive proof protocol for graph non-isomorphism \cite{GS86}. 
\end{itemize}

\section{Algorithms for Tensor Isomorphism}\label{sec:unitaryIso}

We begin by establishing notation used throughout this work.

We reserve $A$ and $B$ to denote $n\times n \times n$ tensors, and $L, R, T$ denote unitary and orthogonal $n\times n$ matrices, unless stated otherwise.
For $L, R, T \in \Unitary(n)$, write
\[
\left((L,R,T)\curvearrowright A\right)_{ijk} = \sum_{p,q,r}L_{ip}R_{jq}T_{kr}A_{pqr}.
\]
In order to avoid ambiguous notation, we let $\Orb(A)$ denote the \emph{orbit} of a tensor $A$, we let $\Unitary(n)$ denote the $n$-dimensional unitary group, $\Orth(n)$ denote the $n$-dimensional orthogonal group, and we let $O(\cdot)$ denote the big-``O" function. As such, $\Orb_\CC(A) \coloneqq \Orb_{\Unitary(n)^3}(A)$ denotes the orbit of $A$ by action of $\Unitary(n)^3$. Similarly, $\Orb_\RR(A) \coloneqq \Orb_{\Orth(n)^3}(A)$.

Let $\QQ(i) = \{s + it : s, t \in \QQ\}$ denote the field of Gaussian rationals, taken to be encoded in the usual way---for the usual floating point representation they take form $(s+it)2^p$ with $s, t \in \QQ$ and $p \in \ZZ$ encoded in binary in the standard way. Let $S^1 = \{z \in \CC^\times : |z|=1\}$. For $v$ in a vector space $V$, write
\[
\supp(v) \coloneqq \{j \in [n]: v_j \neq 0\}.
\]

When $A \in \CC^{n_1\times n_2\times \cdots \times n_k}$, the \emph{mode-$d$ flattening of $A$} is a matrix
\[
A_{d} \in \CC^{n_d\times (n_1\cdots n_{d-1}n_{d+1}\cdots n_k)}.
\]
Let $G_d(A)$ denote the Gram matrix $A_{d}A_{d}^*$; notice these are Hermitian and PSD, and hence have nonnegative real eigenvalues.
In this paper, the mode-$d$ flattening of $A \in \CC^{n\times n \times n}$ is $A \in \CC^{n \times n^2}$ and $G_d(A)$ is $n \times n$.

We let $\AA$ denote the algebraic numbers. There is a well-developed theory for the effective and efficient manipulation of algebraic numbers, see \cref{sec:algebraicNumbers} for discussion. One important comment is that for a matrix $M \in \AA^{n\times n}$, we can compute the Jordan Normal Form of $M$ exactly over the algebraic numbers using, for example, the polynomial-time algorithm of Cai \cite{cai1994computing}.

Let $\| \cdot \|$ denote the spectral norm of a matrix, and $\| \cdot \|_F$ the Frobenius norm.

\subsection{Exact unitary tensor isomorphism}

The purpose of this section is to prove the following theorem, which implies \cref{thm:exactUnitIso}.

\begin{theorem}\label{thm:UnitIsoPrelim}
Let $A \in\QQ(i)^{n\times n\times n}$ have entries drawn independently from some fixed subgaussian distribution, and let $B \in \QQ(i)^{n\times n \times n}$ be arbitrary.
There is a deterministic algorithm that on input $A, B$, decides with high probability (probability $1-\frac{1}{\poly(n)}$) whether
\[
B\in\Orb_\CC(A)\coloneqq \{(L,R,T)\curvearrowright A:\ L,R,T\in \Unitary(n)\} 
\]
in polynomial time, and outputs such $(L,R,T)$ in the affirmative case.
\end{theorem}

A crucial step in the algorithm used to prove \cref{thm:UnitIsoPrelim} is a reduction to deciding the orbit equality problem for actions by compact tori; as a consequence we use the results of B\"urgisser, Do\u{g}an, Makam, Walter, and Wigderson in \cite{burgisser_et_alCCC.2021.32}.

The primary purpose of \cite{burgisser_et_alCCC.2021.32} is to study the orbit equality, orbit closure intersection, and orbit closure containment problems in the context of actions of commutative groups, specifically tori, using techniques from invariant theory. Their main result is to develop polynomial-time algorithms deciding these orbit problems for torus actions. 
A critical insight the authors of \cite{burgisser_et_alCCC.2021.32} make is to notice that in the case of torus actions, a setting presents itself that leaves the problems vulnerable to a technique based on a result of Mumford (see \cite{mumford1994geometric}). Namely, if the orbit closures of vectors $v, w$ of a vector space $V$ do not intersect, then they can always be distinguished by an invariant polynomial.\footnote{A polynomial function $f$ on $V$ is \emph{invariant} if it is constant along orbits, so $f(gv) = f(v)$ for all $g \in G$ and $v \in V$, recalling further that we are considering (linear) actions of $G$ on $V$, using a representation $\rho:G\rightarrow \text{GL}(V)$ and the orbits (notating $\rho(g)v$ by $gv$) $\Orb(v) = \{gv:g \in G\}$. Mumford's result then says that if $\overline{\Orb(v)}\cap \overline{\Orb(w)} = \emptyset$, then there is an invariant polynomial $f$ such that $f(v)\neq f(w)$.}
In particular, and most importantly, in the case of torus actions the authors craft an approach to constructing an arithmetic circuit with division whose output gates compute a system of generating Laurent monomials (rational invariants), which can be used to efficiently check for separation of orbit closures. This forms the heart of their algorithms, and gives an explicit separating invariant.

The main result of \cite{burgisser_et_alCCC.2021.32} we use below is

\begin{proposition*}[{\cite[Corollary 1.6]{burgisser_et_alCCC.2021.32}}]
    Let the weight matrix $M \in \text{Mat}_{d, n}(\ZZ)$ define an $n$-dimensional representation of $\cT = (\CC^\times)^d$ and put $\cK = (S^1)^d$. Further, let $v, w \in \QQ(i)^n$ and assume that the bit-lengths of the entries of $v, w$ and $M$ are bounded by $b$. Then, in $\poly(d, n, b)$-time, we can decide if $\Orb_{\cK}(v) = \Orb_\cK(w)$.
\end{proposition*}

We will also be implicitly using

\begin{proposition*}[{\cite[Proposition 8.1]{burgisser_et_alCCC.2021.32}}]
    Let $M \in \text{Mat}_{d, n}(\ZZ)$ define an $n$-dimensional representation of $\cT = (\CC^\times)^d$ and $\cK = (S^1)^d$. Let $v, w \in \CC^n$. Then $\Orb_\cK(v) = \Orb_\cK(w)$ if and only if $\Orb_\cT(v) = \Orb_\cT(w)$ and $|v_j| = |w_j|$ for all $j$.
\end{proposition*}

We now present the primary algorithm of this paper, used to establish \cref{thm:UnitIsoPrelim}. We first give some additional notation and conventions. Put $\cK = (S^1)^{3n}$. Then for $(a, b, c) \in \cK$ we obtain the diagonal unitary matrices
\[
D_1 = \diag(a_1,\dots,a_n),\; D_2 = \diag(b_1,\dots,b_n), \; D_3 = \diag(c_1,\dots,c_n),
\]
so $\cK$ acts on the vectorizations $v_A$ and $v_B$ of tensors $A$ and $B$ by
\[
(D_3\otimes D_2 \otimes D_1) v_A \text{ and } (D_3\otimes D_2 \otimes D_1) v_B.
\]
Specifically, coordinatewise, the $(i,j,k)$ entry gets multiplied by $a_ib_jc_k$.

\begin{algorithm}[H]
\caption{Deciding unitary equivalence of 3-tensors with simple spectrum.}
\label{alg:1}
    \begin{enumerate}[label=Step \arabic*]
    \item [\textbf{Input}]  On input $A, B$ as well as $\Lambda^A_d, U_d, \Lambda_d^B, V_d$, $d=1,2,3$, we aim to decide in time polynomial in the input size (measured by the bit-lengths of the inputs) whether there are $L,R,T$ in $\Unitary(n)$ such that $(L,R,T)\curvearrowright A=B$. Suppose, without any loss of generality, that the eigenvalues composing $\Lambda_d^A$ and $\Lambda_d^B$ are given in nonincreasing order.
    
    \item\label{alg1step1} For each $d=1,2,3$, check whether, for some $d$, the spectrum of either $G_d(A)$ or $G_d(B)$ is simple by studying the entries matrices $\Lambda_d^A, \Lambda_d^B$. If the spectra are not simple, output \textsf{Cannot decide instance} and halt.
    If the spectra are simple, check whether $\Lambda_d^A = \Lambda_d^B$, $d=1,2,3$. If the spectra are not equal, output \textsf{NO}. If yes, put $\Lambda = \Lambda_d^A = \Lambda_d^B$.
    
    \item\label{alg1step2} Set
    \[
    S_A\coloneqq (U_1^*, U_2^*, U_3^*) \curvearrowright A,\;\text{ and }\;
    S_B\coloneqq (V_1^*, V_2^*, V_3^*) \curvearrowright B.
    \]
    Let $v_A=\vecc(S_A),\,v_B=\vecc(S_B)\in\CC^{N}$ with $N=n^3$, be the vectorizations of $S_A$ and $S_B$, with $(v_A)_l = (S_A)_{ijk}$ by the lexicographic order $l=(k-1)+n(j-1)+n^2(i-1)$ (same for $v_B$).
    
    \item\label{alg1step3} Verify that $\supp(v_A) = \supp(v_B)$, and that for all $l\in [n^3]$, it holds that $|(v_A)_l| = |(v_B)_l|$. If any of these conditions fail, output \textsf{NO}.
    
    \item\label{alg1step4pre} As noted in \cite{burgisser_et_alCCC.2021.32}, their main result extends to the case where the entries of the elements $v \in \CC$ being acted upon are taken from some algebraic number field by a result of Ge \cite{ge1993testing}. We use this fact to apply their algorithm to the case of vectors $v_A$ and $v_B$ taken from the field $\QQ(i, \{\lambda_p\}, \{u_{pq}\}, \{v_{pq}\})$ for $\lambda_p$ the eigenvalues of the Gram matrices $G_d(A)$ and $G_d(B)$ and $u_{pq}$, $v_{pq}$ the entries of $U_d$, $V_d$, $d=1,2,3$.
    \item\label{alg1step4} Using Corollary 1.6 (specifically Algorithm 2) in \cite{burgisser_et_alCCC.2021.32}, 
    check whether $\Orb_\cK(v_A) = \Orb_\cK(v_B)$. If so, output \textsf{YES} and optionally output the witness $L = V_1 D_1 U_1^*$, $R = V_2D_2U_2^*$, $T = V_3D_3U_3^*$, where
    \[
    D_1 = \diag(a_1,\dots,a_n),\; D_2 = \diag(b_1,\dots,b_n), \; D_3 = \diag(c_1,\dots,c_n)
    \]
    for appropriate element $(a, b, c) \in \cK$ found by the algorithm.
    Otherwise output \textsf{NO}.
    \end{enumerate}
\end{algorithm}

Before stating the proof of \cref{thm:UnitIsoPrelim} we require some lemmas.

\begin{lemma}\label{lem:fundamentalCondition}
Let $A,B \in \CC^{n\times n \times n}$. Suppose $G_d(A)$ and $G_d(B)$ have a simple spectrum that is equal up to permutation, for each $d=1,2,3$.  Write the unitary diagonalizations $G_d(A)=U_d\Lambda^A_d U_d^*$ and $G_d(B)=V_d\Lambda^B_d {V_d}^*$, with eigenvalues in nonincreasing order and a phase convention fixed. Define
\[
S_A\coloneqq (U_1^*, U_2^*, U_3^*) \curvearrowright A,\quad
S_B\coloneqq (V_1^*, V_2^*, V_3^*) \curvearrowright B.
\]
Then there exists unitary matrices $L, R, T \in \Unitary(n)$ such that $(L, R, T) \curvearrowright A = B$ if and only if there exist diagonal $D_1,D_2,D_3 \in \Unitary(1)^n$ such that $(D_1,D_2,D_3)\curvearrowright S_A = S_B$.
\end{lemma}
\begin{proof}
We begin by recalling that if $B=(L,R,T)\curvearrowright A$, then
\[
G_1(B)=LG_1(A)L^*,\;
G_2(B)=RG_2(A)R^*,\;
G_3(B)=TG_3(A)T^*;
\]
for example, for $d=1$, $B_{1}=LA_{1}(T\otimes R)^*$, hence $G_1(B)=B_{1}{B_{1}}^*=LA_{1}{A_{1}}^* L^*=LG_1(A)L^*$.

We now proceed with the proof.

($\Longrightarrow$) Put $L_1=L$, $L_2=R$, $L_3=T$, which are taken from $\Unitary(n)$. Recall $G_d(A)=U_d\Lambda^A_d U_d^*$ and $G_d(B)=V_d\Lambda^B_d {V_d}^*$.
Assume $B=(L_1,L_2,L_3)\curvearrowright A$. Then $\Lambda^A_d = \Lambda^B_d = \Lambda_d$, so we have
\[
G_d(B)=V_d\Lambda_d {V_d}^* =L_d G_d(A) L_d^* = L_d U_d\Lambda_d U_d^* L_d^*.
\]
$\Lambda_d$ is simple and ordered, so the ordered eigenbases of $G_d(B)$ are unique up to phase. So there is a $D_d \in U(1)^n$ such that 
\[
V_d = L_d U_d D_d \text{ for } d=1,2,3.
\]
As such,
\begin{align*}
S_B&=(V_1^*,V_2^*,V_3^*)\curvearrowright B\\
&=(V_1^* L_1,V_2^* L_2,V_3^* L_3)\curvearrowright A\\
&=(D_1^* U_1^*,D_2^* U_2^* ,D_3^* U_3^*)\curvearrowright A\\
&=(D_1^*,D_2^*,D_3^*)\curvearrowright S_A.
\end{align*}
Putting $D_d \leftarrow D_d^*$ we obtain the first direction.

($\Longleftarrow$) On the other hand, suppose $(D_1,D_2,D_3)\curvearrowright S_A=S_B$ where $D_d \in U(1)^n$ as before. Put
\[
L\coloneqq V_1 D_1 U_1^*,\quad R\coloneqq V_2 D_2 U_2^*,\quad T \coloneqq V_3 D_3 U_3^*\in \Unitary(n).
\]
Then, using the hypothesis,
\begin{align*}
(L,R,T)\curvearrowright A
&=(V_1D_1,V_2D_2,V_3D_3)\curvearrowright (U_1^*,U_2^*,U_3^*)\curvearrowright A\\
&=(V_1,V_2,V_3)\curvearrowright \big((D_1,D_2,D_3)\curvearrowright S_A\big)\\
&=(V_1,V_2,V_3)\curvearrowright S_B\\
&=B.
\end{align*}
Furthermore, $S_B=(V_1^*,V_2^*,V_3^*)\curvearrowright B$ implies $B=(V_1,V_2,V_3)\curvearrowright S_B$, so we conclude $A$ and $B$ are unitarily equivalent.
\end{proof}

We now prove correctness and efficient run-time of \cref{alg:1}.

\begin{lemma}\label{lem:algoCorrect}
    \cref{alg:1} is correct.
\end{lemma}
\begin{proof}
We prove \cref{alg:1} outputs $\textsf{YES}$ if and only if there exist $L, R, T \in \Unitary(n)$ with $(L, R, T) \curvearrowright A = B$.

First, notice the Gram matrices have entries in $\QQ(i)$, and hence can be diagonalized by extending to the algebraic numbers. Proceeding in any one of many well-known ways (such as Cai's technique \cite{cai1994computing}) to achieve the decompositions $G_d(A) = U_d \Lambda_d U_d^*$ and $G_d(B)=V_d\Lambda_d {V_d}^*$ exactly using methods for handling algebraic numbers effectively.
Simplicity of the spectra of the $G_d(A)$ and $G_d(B)$ is required for the fourth step of the algorithm. And, clearly, if the eigenvalues of $G_d(A)$ and $G_d(B)$ are not equal (up to permutation), then there do not exist $L, R, T \in \Unitary(n)$ such that $(L, R, T) \curvearrowright A = B$.

Once the spectra of the Gram matrices $G_d(A)$ and $G_d(B)$ have been checked, \cref{lem:fundamentalCondition} tells us that it suffices to determine whether the constructed $S_A$ and $S_B$ are in the same orbit (by a compact torus action): $A$ and $B$ are $\Unitary(n)^3$-isomorphic if and only if there exist diagonal matrices $D_1,D_2,D_3 \in \Unitary(1)^n$ such that $(D_1,D_2,D_3)\curvearrowright S_A = S_B$. Vectorizing $S_A$ and $S_B$ in the canonical way mentioned in \ref{alg1step2} has no effect on this, and is simply done for easier application of the algorithm of \cite{burgisser_et_alCCC.2021.32}.
The check in \ref{alg1step3} that the supports and moduli of the $v_A$ and $v_B$ match follows from Proposition 8.1 and the contrapositive of Lemma 3.2 in \cite{burgisser_et_alCCC.2021.32}.
Finally, the algorithm of \cite{burgisser_et_alCCC.2021.32} tells us whether $\Orb_\cK(v_A) = \Orb_\cK(v_B)$.

Furthermore, if the algorithm outputs yes, it may also output $L= V_1D_1U_1^*$, $R = V_2D_2U_2^*$, and $T = V_3D_3U_3^*$. Note first that each $U_d$ and $V_d$ is unitary by construction, and each $D_d$ is diagonal with entries in the complex unit circle, and is hence unitary. Finally, we have
\begin{align*}
(L,R,T)\curvearrowright A
&=(V_1D_1U_1^*,V_2D_2U_2^*,V_3D_3U_3^*) \curvearrowright A\\
&=(V_1,V_2,V_3)\curvearrowright \big((D_1,D_2,D_3)\curvearrowright (U_1^*,U_2^*,U_3^*)\curvearrowright A\big)\\
&=(V_1,V_2,V_3)\curvearrowright \big((D_1,D_2,D_3)\curvearrowright S_A\big)\\
&=(V_1,V_2,V_3)\curvearrowright S_B\\
&=B
\end{align*}
as required.
\end{proof}

\begin{lemma}\label{lem:algoPoly}
    \cref{alg:1} halts in time polynomial in the length of the binary encodings of $A$, $B$.
\end{lemma}
\begin{proof}
The formation of $G_d(A)$ and $G_d(B)$ requires $O(n^4)$ arithmetic operations. Moreover, computing the characteristic polynomials and testing simplicity of the spectrum by extending to the algebraic numbers can all be performed in polynomial time; specifically, computing the Jordan Normal Form, and thus the decompositions of $G_d(A)$ and $G_d(B)$ over the algebraic numbers exactly can be performed in polynomial-time  \cite{cai1994computing}. Hence, $S_A$ and $S_B$ can be efficiently constructed, with their vectorizations as well. Testing support and modulus equality in \ref{alg1step3} involves at most $n^3$ checks, and each check is a comparison of algebraic numbers which is a polynomial-time operations in the lengths of their representations (see \cref{sec:algebraicNumbers}).

For \ref{alg1step4}, we may construct an integer weight matrix $M\in\mathbb{Z}^{3n\times N}$ whose column for $(i,j,k)$ is $e_i\oplus e_{n+j}\oplus e_{2n+k}$; the algorithm for checking orbit equality then runs in polynomial time by Corollary 1.6 of \cite{burgisser_et_alCCC.2021.32}, and remains polynomial-time when extending to the different algebraic number fields \cite{burgisser_et_alCCC.2021.32,ge1993testing}.

We conclude that the algorithm halts in time polynomial in the length of the input.
\end{proof}

With these lemmas in hand we are in a position to prove \cref{thm:UnitIsoPrelim}.

\begin{proof}[Proof of \cref{thm:UnitIsoPrelim}]
    \cref{lem:algoCorrect,lem:algoPoly} prove correctness and polynomial runtime of \cref{alg:1}. By \cref{thm-repulsion-unscaled} the Gram matrices $G_d(A)$ and $G_d(B)$ have simple spectrum with high probability; this gives the fact that the algorithm decides whether $A$ and $B$ are unitarily isomorphic with high probability.
\end{proof}

\begin{remark}
Perhaps it is not obvious why \cref{alg:1} requires the Gram matrices $G_d(A), G_d(B)$ to have simple spectrum. Simple spectrum is necessary to invoke \cite{burgisser_et_alCCC.2021.32} after verifying eigenbasis alignment: any further freedom is in per-coordinate phases (a torus action). If the spectrum are not simple, the residual freedom inside each eigenspace of multiplicity $m$ is the full unitary group $\Unitary(n)$, not just the phases. In such a case a different algorithm is required.
\end{remark}

\begin{remark}
    \ref{alg1step4} of \cref{alg:1} invokes rather heavy machinery from \cite{burgisser_et_alCCC.2021.32} to decide whether $A$ and $B$ are isomorphic. On the other hand, in \cref{sec:RobustOrthogonalProblem} we present a related algorithm (in the fixed precision case, which can be readily adapted to the perfect precision setting), which does not need the torus action machinery of \cite{burgisser_et_alCCC.2021.32} but instead uses simple linear programming to solve a problem related to what is achieved in \ref{alg1step4} of \cref{alg:1}. It would appear that a similar, albeit more subtle method can be used to replace the invocation of the torus action results of \cite{burgisser_et_alCCC.2021.32} in \ref{alg1step4}.
\end{remark}

\subsection{Approximate unitary tensor isomorphism}

In this section we move to an approximate, numerical setting for studying orbit distance problems rather than exact orbit equality problems as in the previous section.
We we start by stating some results concerning perturbations of matrices and tensors.

Recall that the trace norm of a matrix $A \in \CC^{n\times n}$ is given by
\[
\|A\|^2 = \mathrm{Tr}(AA^*) = \mathrm{Tr}(A^* A) = \sum_{1\leq i, j \leq n} |A_{i,j}|^2.
\]
We write $\|A\|_F = \sqrt{\|A\|^2}$ for the Frobenius norm. Clearly, we have the following upper bound on the Frobenius norm:
\[
\|A\|_F \leq n \max_{i,j} |a_{i,j}|
\]
for $A \in \CC^{n\times n}$.

We now recall a useful inequality of Hoffman and Wielandt \cite{hoffman1953variation}. The Hoffman-Wielandt inequality asserts that if $A, B \in \CC^{n\times n}$ are normal, with respective eigenvalues $\lambda_1(A),\dots,\lambda_n(A)$ and $\lambda_1(B),\dots,\lambda_n(B)$, then, letting $\cS_n$ denote the permutation group of $\{1,\dots,n\}$, we have
\[
\min_{\sigma \in \cS_n}\sqrt{\sum_{i=1}^n |\lambda_i(A)-\lambda_{\sigma(i)}(B)|^2} \leq \|A-B\|_F.
\]
Weyl and Mirsky also provide related useful inequalities for singular values (see \cite{stewart1998perturbation}).
We need the following

\begin{lemma}\label{lem:truncation}
    Let $A,\tilde{A} \in \RR^{n\times n}$ be symmetric, and let $\varepsilon = \max_{i,j}|A_{i,j}-\tilde{A}_{i,j}|$.  Then
    \[
    \min_{\sigma \in \cS_n}\sqrt{\sum_{i=1}^n |\lambda_i(A)-\lambda_{\sigma(i)}(\tilde{A})|^2} \leq n\varepsilon.
    \]
\end{lemma}
\begin{proof}
If $A, \tilde{A} \in \RR^{n\times n}$ are real symmetric, then they are normal; the inequality of Hoffman and Wielandt applies. So, using the fact that $\|A\|_F \leq n \max_{i,j} |a_{i,j}|$ for $A \in \R^{n\times n}$, the assertion follows.
\end{proof}

\begin{remark}
    Due to a theorem of Mirsky \cite{mirsky1960symmetric}, we have the following: if $\tilde{A} = A + E$ is a perturbation of an arbitrary matrix $A$, then
    \[
    \min_{\sigma \in \cS_n}\sqrt{\sum_{i=1}^n |\mu_i-\tilde{\mu}_{\sigma(i)}|^2} \leq \|E\|_F,
    \]
    Where the $\mu_i, \tilde{\mu}_i$ are the singular values of $A$ and $\tilde{A}$, respectively.
    But, since the eigenvalues and singular values coincide for real symmetric matrices, this means the prior lemma holds in a more general setting: we do not actually need to assume that the perturbed matrix $\tilde{A}$ is symmetric.
\end{remark}

The use of these results is that, in a numerical setting where the entries of a matrix are truncated to a fixed precision, we conclude that the eigenvalues and singular values of the matrix cannot move too much through the truncation operation. As a consequence, lower-bounds on the gaps of eigenvalues carry over to the fixed-precision setting.

We give one useful corollary.

\begin{corollary}
    Let $A, B \in \CC^{n\times n}$ be real symmetric matrices with a simple spectrum.
    Put 
    \[
    \delta_{max} = \max_{1\leq i \leq n} \min_{\sigma \in \cS_n} |\lambda_i(A)-\lambda_{\sigma(i)}(B)|.
    \]
    Then there exists a unitary $L \in \CC^{n\times n}$ such that
    \[
    \|A-LBL^*\|_F \leq n \delta_{max}.
    \]
\end{corollary}
\begin{proof}
    The matrices $A, B$ are real symmetric with a simple spectrum, so we can write $A = U \Lambda U^*$ and $B = V \Sigma V^*$ for $U, V$ unitary. Put $L = UV^*$, and put $\delta_i = \min_{\sigma \in \cS_n} |\lambda_i(A)-\lambda_{\sigma(i)}(B)|$.
    Then $A-LBL^* = U(\Lambda-\Sigma)U^*$. Because the Frobenius norm $\|\cdot\|_F$ is unitarily invariant,
    \[
    \|A-LBL^*\|_F = \|U(\Lambda-\Sigma)U^*\|_F = \|\Lambda-\Sigma\|_F=  \sqrt{\sum_{i=1}^n\delta_i^2} \; \leq \; n \delta_{max}.\qedhere
    \]
\end{proof}

The following lemma is useful for translating problems involving exact orbit equality or closure intersections for tensors into their approximate (or numerical) counterparts.
For $X \in \CC^{n\times n\times n}$, Let $\|X \|_\infty$ denote the entrywise $\ell_\infty$ norm, i.e. $\max_{1\leq i,j,k \leq n} |X_{ijk}|$.

\begin{lemma}\label{lem:perturbedOrbit}
Let $A, B \in \CC^{n\times n \times n}$, and suppose there are unitary $L, R, T$ such that $(L, R, T) \curvearrowright A = B$. Fix $\varepsilon_0 > 0$. Then for all $\tilde{A} \in \CC^{n \times n\times n}$ satisfying $\|A - \tilde{A}\|_\infty \leq \varepsilon_0$,
\[
\|(L, R, T) \curvearrowright \tilde{A} - B\|_F \leq n^{3/2}\varepsilon_0.
\]
\end{lemma}
\begin{proof}
Put $\Delta = A - \tilde{A}$. Notice
\begin{align*}
    (L, R, T) \curvearrowright \tilde{A} - B &= (L,R,T) \curvearrowright (A + \Delta) - (L, R, T)\curvearrowright A\\
    &= (L, R, T) \curvearrowright \Delta.
\end{align*}
The action is unitary on $\CC^{n \times n \times n}$ with respect to the Frobenius inner product, so it preserves the norm. Hence
\[
\|(L, R, T) \curvearrowright \Delta\|_F = \|\Delta\|_F.
\]
By the assumption that $\|A - \tilde{A}\|_\infty \leq \varepsilon_0$, we have $\|\Delta\|_\infty < \varepsilon_0$ and $|\Delta_{ijk}|\leq \varepsilon_0$ for all $i,j,k$. But
\[
\|\Delta\|_F^2 = \sum_{i,j,k}|\Delta_{ijk}|^2 \leq \sum_{i,j,k}\varepsilon_0^2 = n^3\varepsilon^2
\]
implying $\|\Delta\|_F \leq n^{3/2}\varepsilon$. Combining the above observations we see that
\[
\|(L, R, T) \curvearrowright \tilde{A} - B\|_F \leq n^{3/2}\varepsilon_0. \qedhere
\]
\end{proof}

Using these observations and \cref{thm:UnitIsoPrelim}, we conclude that if there are unitary matrices $L, R, T$ such that $B = (L, R, T) \curvearrowright \tilde{A}$, where $\tilde{A}$ denotes the truncation of a tensor $A$ with complex entries, then
\[
    \|(L, R, T)\curvearrowright A - B\|_F \leq \frac{n^{3/2}}{2^\ell}.
\]
This follows from \cref{lem:perturbedOrbit}, because if $(L, R, T) \curvearrowright \tilde{A} = B$, then for all $A \in \CC^{n\times n \times n}$ satisfying $\|A - \tilde{A}\|_\infty \leq \varepsilon_0$, $\|(L,R,T)\curvearrowright A - B\|_F \leq n^{3/2}\varepsilon_0$. In our setting $\varepsilon_0 = 1/2^\ell$, so the assertion follows.

This observation leads nicely into the sequel.

\paragraph{Approximate Unitary TI.}
Using numerical algorithms rather than perfect-precision techniques as in \cref{alg:1} introduces errors that make checking (compact) orbit equality problems infeasible. Instead, one could hope to decide whether the orbits of two 3-tensors by unitary action are ``close" or ``far" apart.
The content of the following theorem is to show that in the fixed-precision setting, there is a numerical algorithm for unitary 3-tensor isomorphism with the property that ``yes" instances imply an upper-bound on the distance
\[
\dist_\CC(A, B) \coloneqq \inf_{L, R, T \in \Unitary(n)} \|(L, R, T) \curvearrowright A - B\|_F
\]
when $A$ and $B \in \QQ(i)^{n\times n \times n}$.

More specifically, the following theorem proves existence of an algorithm solving what we shall call the \emph{approximate unitary tensor isomorphism problem.} 

\begin{problem}\label{prob-3}
    The \emph{approximate unitary tensor isomorphism problem} is the following: given input of an operating precision $\ell$, and tensors $A$ and $B$ in $\QQ(i)^{n\times n \times n}$ with entries specified to $\ell$ bits of precision, decide between the following cases. 
    \begin{enumerate}
        \item Output ``no'' if there are no orthogonal matrices $L, R, T$ such that
\[
(L, R, T) \curvearrowright A = B.
\]
\item Otherwise, attempt to construct a fixed-precision witness $(\widetilde{L}, \widetilde{R}, \widetilde{T}) \in \Unitary(n)^3$ such that
\[
\|(\widetilde{L}, \widetilde{R}, \widetilde{T})\curvearrowright A - B\|_F\leq O\left(\frac{n}{2^\ell}\right),
\]
outputting ``yes" if such a witness can be found.
    \end{enumerate}
\end{problem}

\begin{theorem}\label{thm:numericalUnitIso}
There exists a deterministic polynomial-time algorithm which, given $A\in\QQ(i)^{n\times n\times n}$ with i.i.d.\ entries $a_{ijk} \sim \xi$ for a fixed---mean zero, unit variance---sub-Gaussian distribution $\xi$ and $B\in\QQ(i)^{n\times n\times n}$ arbitrary and deterministic, with high probability decides the approximate unitary tensor isomorphism problem for $A$ and $B$, with implied constants depending only on the distribution of $\xi$.
\end{theorem}

We establish \cref{thm:numericalUnitIso} through a straight-forward adaptation of \cref{alg:1}. The key difference between \cref{alg:1} and \cref{alg:2} presented below is in the numerical computation of the decompositions of the Gram matrices $G_d(A)$, $G_d(B)$, $d=1,2,3$, rather then computing the decompositions perfectly using algebraic numbers. This step introduces errors into the procedure that preclude the ability to decide whether $\Orb_{\CC}(A) = \Orb_{\CC}(B)$, and instead force us to conclude orbit closeness.

The modified numerical algorithm is presented below in \cref{alg:2}.

\begin{algorithm}[H]
\caption{Compact orbit distance decision procedure for unitary actions.}
\label{alg:2}
    \begin{enumerate}[label=Step \arabic*]
    \item [\textbf{Input}]  Let the positive integer $\ell$ denote the operating precision. Take as input $A, B \in \QQ(i)^{n\times n \times n}$ (with $\ell$ the maximum bit-length of the entries of $A$ and $B$ in the standard binary encoding).
    
    \item\label{alg2Step1}
    Let $G_d(A)$ and $G_d(B)$, $d=1,2,3$ denote the Gram matrices obtained from the mode-$d$ flattenings of $A$ and $B$, respectively, and using a numerical algorithm with good forward error bounds (such as that proposed in \cite{dey2023bit}), to compute the unitary diagonalizations (JNFs) $\widetilde{U_d} \widetilde{\Lambda_d^A} \widetilde{U_d}^*$ of $G_d(A)$, and $\widetilde{V_d} \widetilde{\Lambda_d^B} \widetilde{V_d}^*$ of $G_d(B)$, to $\ell$ bits of precision per entry of the matrices.
    
    Specifically, if $\|\cdot\|$ denotes the operator norm of a matrix and $\|\cdot\|_\infty$ denote the entrywise infinity norm, recalling $\|M\|_\infty\leq \|M\|\leq n\|M\|_\infty$; then Algorithm 1 of \cite{dey2023bit} gives a polynomial-time procedure giving $\widetilde{U_d} \widetilde{\Lambda_d^A} \widetilde{U_d}^*$ and $\widetilde{V_d} \widetilde{\Lambda_d^B} \widetilde{V_d}^*$ so that
    \[
    \|U_d-\widetilde{U_d}\| \leq \frac{n\|U_d\|_\infty}{2^\ell} \text{ and } \|\Lambda_d^A-\widetilde{\Lambda_d^A}\| \leq \frac{n\|\Lambda_d^A\|_\infty}{2^\ell}
    \]
    for the exact $G_d(A) = U_d\Lambda_d^A U_d^*$. The same holds for the decomposition of $G_d(B)$.
    
    Suppose, without any loss of generality, that the eigenvalues composing $\widetilde{\Lambda_d^A}$ and $\widetilde{\Lambda_d^B}$ are given in nonincreasing order.
    For each $d=1,2,3$, check whether, for some $d$, the spectrum of either $G_d(A)$ or $G_d(B)$ is simple by studying the entries matrices $\widetilde{\Lambda_d^A}, \widetilde{\Lambda_d^B}$. If either $\widetilde{\Lambda_d^A}$ or $\widetilde{\Lambda_d^B}$ have non-simple spectrum, output \textsf{CANNOT DECIDE INSTANCE} and halt.
    If the spectra are simple, check whether $\widetilde{\Lambda_d^A} = \widetilde{\Lambda_d^B}$ for each $d=1,2,3$. If for some $d$ we have $\widetilde{\Lambda_d^A} \ne \widetilde{\Lambda_d^B}$, output \textsf{NO}.
    
    \item\label{alg2Step2} Set
    \[
    \widetilde{S_A}\coloneqq (\widetilde{U_1}^*, \widetilde{U_2}^*, \widetilde{U_3}^*) \curvearrowright A,\;\text{ and }\;
    \widetilde{S_B}\coloneqq (\widetilde{V_1}^*, \widetilde{V_2}^*, \widetilde{V_3}^*) \curvearrowright B.
    \]
    Let $v_A=\vecc(\widetilde{S_A}),\,v_B=\vecc(\widetilde{S_B})\in\CC^{N}$ with $N=n^3$, be the vectorizations of $S_A$ and $S_B$. In particular, set $(v_A)_l = (\widetilde{S_A})_{ijk}$ using the lexicographic order $l=(k-1)+n(j-1)+n^2(i-1)$ (same for $v_B$).
    
    \item\label{alg2Step3} Verify that (i) $\supp(v_A) = \supp(v_B)$, and (ii) that for all $l\in [n^3]$, it holds that $|(v_A)_l| = |(v_B)_l|$. If any of these conditions fail, output \textsf{NO}.
    
    \item\label{alg2Step4} We now apply the results for compact torus actions of \cite{burgisser_et_alCCC.2021.32}: using Corollary 1.6 (specifically Algorithm 2) in \cite{burgisser_et_alCCC.2021.32}, 
    Check whether $\Orb_\cK(v_A) = \Orb_\cK(v_B)$. If so, output \textsf{YES} and optionally output the witness $\widetilde{L} = \widetilde{V_1} \widetilde{D_1} \widetilde{U_1}^*$, $\widetilde{R} = \widetilde{V_2}\widetilde{D_2}\widetilde{U_2}^*$, $\widetilde{T} = \widetilde{V_3}\widetilde{D_3}\widetilde{U_3}^*$, where
    \[
    \widetilde{D_1} = \diag(a_1,\dots,a_n),\; \widetilde{D_2} = \diag(b_1,\dots,b_n), \; \widetilde{D_3} = \diag(c_1,\dots,c_n)
    \]
    for appropriate element $(a, b, c) \in \cK$ (recalling $\cK = (S^1)^{3n}$) found by the algorithm.
    Otherwise output \textsf{NO}.
    \end{enumerate}
\end{algorithm}

Proof of correctness and efficient run-time of \cref{alg:2} follows identical lines with the proof of correctness and efficient run-time of \cref{alg:1}.
We are now in a position to prove \cref{thm:numericalUnitIso}.

\begin{proof}[Proof of \cref{thm:numericalUnitIso}]
       Proof of correctness and efficient run-time of \cref{alg:2} follows from \cref{lem:algoCorrect,lem:algoPoly}, \textit{mutatis mutandis}.
       By \cref{thm-repulsion-unscaled}, with high probability the Gram matrix $G_d(A)$ has simple spectrum, and, crucially the minimal distance between its eigenvalues is greater than $O(1/2^\ell)$; this enables the algorithm to distinguish the eigenvalues at precision $\ell$, and to hence output $\textsf{YES}$ or $\textsf{NO}$ with high probability.

       Finally, suppose the algorithm outputs \textsf{YES} with a witness $(\widetilde{L}, \widetilde{R}, \widetilde{T})$. The forward error bounds stated in \ref{alg2Step1} ensure $\|\widetilde{U}_d - U_d\| \leq O(n 2^{-\ell})$ and similarly for $\widetilde{V}_d$. Propagating this through the action shows that the produced witness $(\widetilde{L},\widetilde{R},\widetilde{T})$ satisfies
       \[
       \|(\widetilde{L},\widetilde{R},\widetilde{T}) \curvearrowright A-B \|_F \leq \left( \|\widetilde{L} - L\| + \|\widetilde{R} - R\| + \|\widetilde{T} - T\| \right)\|A\|_F + O(2^{-2\ell}) \|A\|_F
       \]
       where $(L, R, T)$ is an exact witness. But $\|\widetilde{L}- L\|\leq O(n2^{-\ell})$ (as with $R, T$) gives the concrete bound
       \[
       \|(\widetilde{L},\widetilde{R},\widetilde{T}) \curvearrowright A - B\|_F \leq O\left(\frac{n}{2^\ell}\right). \qedhere
       \]
\end{proof}

\subsection{Unitary orbit approximation}\label{sec:RobustOrthogonalProblem}
Now we zoom out even further to the robust variation of the tensor isomorphism problem. 
Recall the (gapped) unitary tensor orbit approximation problem 
(\cref{prob-2})
from the introduction:
\begin{problem*}[\cref{prob-2}, restated]
The \emph{(gapped) unitary tensor orbit approximation problem} is defined as follows: Given $0<\varepsilon<\eta$ and tensors  $A,B \in \C^{n\times n \times n}$, distinguish between the following two cases:
\[
\dist(\Orb_\C(A), \Orb_\C(B)) \leq \varepsilon \;\text{ and }\; \dist(\Orb_\C(B), \Orb_\C(B)) \geq \eta.
\]
\end{problem*}

Note that \cref{prob-2} differs from \cref{prob-3}, in that it decides between close enough (orbit distance $\leq \epsilon$) and quite far (orbit distance $\geq \eta$). On the other hand, \cref{prob-3} decides between isomorphic (orbit distance $=0$) and non-isomorphic (orbit distance $>0$). 

We present an algorithm (\cref{thm:robustOrthoOrbit,alg:approximateOrbitDistance} below) which solves the unitary tensor orbit approximation problem \cref{prob-2} in the case where the Gram matrices of $A$, $G_d(A)$ for $d\in [3]$ have simple spectrum with a minimum eigenvalue gap of at least $\delta>0$.
\cref{thm-repulsion-unscaled} ensures that the random $A$ provided in \cref{thm:robustOrbit} 
satisfies the aforementioned criteria; this fact paired with the following theorem establishes \cref{thm:robustOrbit}.

\begin{theorem}[Deciding the unitary orbit approximation problem]\label{thm:robustOrthoOrbit}
Let $A,B\in\C^{n\times n\times n}$ such that $G_d(A)$ has simple spectrum with minimum
eigenvalue gap at least $\delta > 0$, for every $d\in [3]$. Then for every $\varepsilon$ satisfying $\delta/4(\|A\|_F+\|B\|_F)> \varepsilon > 0$, there exists an algorithm that runs in time $O(\poly(n, \log(1/\varepsilon))$
distinguishing $\dist_\C(A,B)<\varepsilon$ from $\dist_\C(A,B)>\gamma\varepsilon$  
for $\gamma = O(n^{7/2}\|A\|^2_F/\delta)$, and in the former case outputs a witness $(\widetilde{L},\widetilde{R},\widetilde{T})\in \Unitary(n)^3$ such that 
$\|(\widetilde{L},\widetilde{R},\widetilde{T})\curvearrowright A-B\|_F< \gamma\varepsilon$.
\end{theorem}
\begin{remark}
We make no effort to optimize constants, including the constant exponents of poly$(n)$ terms. 
\end{remark}

\begin{algorithm}[H]
\caption{Orbit distance decision procedure for orthogonal actions.}
\label{alg:approximateOrbitDistance}
\begin{enumerate}[label=Step \arabic*]
\item[\textbf{Input}] Take $A,B\in\C^{n\times n\times n}$ and let $\widetilde{A},\widetilde{B}$ denote the truncations of tensors $A,B$ which we take to $\ell = \lceil\log(1000n^{7}/\varepsilon)\rceil$-bit precision. Let $\gamma = O(n^{7/2}\|A\|^2_F/\delta)$. Recall the \emph{mode-$d$ flattening of a 3-tensor $W$} is denoted $W_{d} \in \C^{n\times (n^2)}$. Recall the Gram matrix $G_d(A)$ is defined as $A_{d}A_{d}^*$, and similarly for $G_d(B)$.

\item\label{alg:magnitude} Check that $\left|\|\widetilde{A}\|_F-\|\widetilde{B}\|_F\right|<2\varepsilon$, and output \textsf{NO} (i.e.\ $\dist_\C(A,B)>\gamma\varepsilon$) if not.

\item\label{alg:stepJNF} Write the Gram matrices $G_d(A)=U_d\Lambda^A_d U_d^*$ and $G_d(B)=V_d\Lambda^B_d {V_d}^*$, with eigenvalues in non-increasing order.
Let $K\coloneqq2\|\widetilde{A}\|_F$.
We compute the unitary diagonalizations of the Gram matrices of $\widetilde{A},\widetilde{B}$ approximately, getting $\widetilde{G_d(A)}=\widetilde{U_d}\widetilde{\Lambda^A_d}\widetilde{U_d^*}$, and $\widetilde{G_d(B)}=\widetilde{V_d}\widetilde{\Lambda^B_d}\widetilde{V_d^*}$ so that 
\begin{equation}
\|\widetilde{U_d}-U_d\|_F,\|\widetilde{V_d}-V_d\|_F,\|\widetilde{\Lambda^A_d}-\Lambda^A_d\|_F,\|\widetilde{\Lambda^B_d}-\Lambda_d\|_F < \frac{\varepsilon}{1000n^7\|A\|_F }\label{eqn:precisionBound}
\end{equation}
using Algorithm 1 in~\cite{dey2023bit}\footnotemark{} and put the eigenvalues in non-increasing order by conjugating the Gram matrices by appropriate permutation matrices.  
$\widetilde{\Lambda^B_d}$ has minimum eigenvalue gap less than $\delta/2$, \textbf{output} \textsf{NO} ($\dist_{\CC}(A,B)>\gamma\varepsilon$) and halt.

\item\label{alg:HOSVDcore} Define
\[
\widetilde{S_A}\coloneqq \left(\widetilde{U_1^*}, \widetilde{U_2^*}, \widetilde{U_3^*}\right) \curvearrowright \widetilde{A},\qquad
\widetilde{S_B}\coloneqq\left( \widetilde{V_1^*}, \widetilde{V_2^*}, \widetilde{V_3^*}\right) \curvearrowright \widetilde{B}.
\]
If there exist $(i,j,k)$ such that $\left||\widetilde{S_A}(i,j,k)|-|\widetilde{S_B}(i,j,k)|\right|>2\varepsilon n^2\delta^{-1}(\|\widetilde{A}\|_F+\|\widetilde{B}\|_F)$, \textbf{output} \textsf{NO} ($\dist_{\CC}(A,B)>\gamma\varepsilon$).
Otherwise for each $(i,j,k)$ with $|\widetilde{S_A}(i,j,k)|+|\widetilde{S_B}(i,j,k)|> 2\varepsilon n^2\delta^{-1}(\|\widetilde{A}\|_F+\|\widetilde{B}\|_F)$, define
\[
\Phi_{ijk} \coloneqq \arg\left(\frac{\widetilde{S_B}(i,j,k)}{\widetilde{S_A}(i,j,k)}\right)
\]

\item\label{alg:feasibilityCheck} Introduce unknowns $\alpha_i,\beta_j,\gamma_k\in[0,2\pi)$ (encoding phases $s_1(i)=e^{i\alpha_i}$, $s_2(j)=e^{i\beta_j}$, $s_3(k)=e^{i\gamma_k}$) and the linear inequalities
\[
|\Phi_{ijk}-(\alpha_i+\beta_j+\gamma_k)|\pmod{2\pi} < \arccos\left(\frac{|\widetilde{S_A}(i,j,k)|^2+|\widetilde{S_B}(i,j,k)|^2-2\varepsilon^2 n^4\delta^{-2}(\|\widetilde{A}\|_F+\|\widetilde{B}\|_F)^2}{2\; |\widetilde{S_A}(i,j,k)|\;|\widetilde{S_B}(i,j,k)|}\right)\]\[\quad\text{for all }(i,j,k)\text{ with } \Phi_{ijk} \text{ defined above.}
\]

\begin{itemize}
\item If \emph{infeasible}, output \textsf{NO} ($\dist_\C(A,B) > \gamma\varepsilon$).
\item If \emph{feasible}, construct $D_1=\mathrm{diag}(e^{-i\alpha_i})$, $D_2=\mathrm{diag}(e^{i\beta_j})$, $D_3=\mathrm{diag}(e^{i\gamma_k})$ and set
\[
\widetilde{L}\coloneqq\widetilde{V_1}D_1\widetilde{U_1^*},\quad \widetilde{R}\coloneqq\widetilde{V_2}D_2\widetilde{U_2^*},\quad \widetilde{T}\coloneqq\widetilde{V_3}D_3\widetilde{U_3^*}.
\]
output \textsf{YES} ($\dist_\C(A,B) < \varepsilon$) and (optionally) the witness $(\widetilde{L},\widetilde{R},\widetilde{T})$ (note that these witnesses only attest to the fact that $\dist_\C(A,B) < \gamma\varepsilon$).
\end{itemize}
\end{enumerate}

\end{algorithm}
\footnotetext{Note that we are not using the full power of this algorithm}

Before we prove correctness and efficient run-time of our algorithm, we provide some lemmas.

\begin{lemma}\label{lem:traceBound}
Let $A,B\in \C^{n\times n\times n}$ with $G_d(A)=U_d\Lambda^A_d U_d^*$ and $G_d(B)=V_d\Lambda^B_d {V_d}^*$ for $d\in [3]$, 
with $\Lambda_d^A,\Lambda_d^B$ ordered with non-increasing entries on the diagonal, let $\varepsilon>0$
and let $\delta>0$ denote the 
minimum eigenvalue gap of $G_d(A)$ over all $d\in [3]$.

If $G_d(B)$ has some eigenvalue gap $\lambda_i-\lambda_j < \delta-2\varepsilon$ for some $d$ and $i,j$, then
\[\dist_\C(A,B)> \frac{\varepsilon}{(\|A\|_F+\|B\|_F)}.\]
\end{lemma}

\begin{proof}
Let
\[
S_A\coloneqq \left(U_1^*, U_2^*, U_3^*\right) \curvearrowright A,\qquad
S_B\coloneqq\left( {V_1^*}, {V_2^*},{V_3^*}\right) \curvearrowright B.
\]
Since $S_A,S_B$ are in the unitary orbits of $A,B$ respectively, 
${\dist_\C(A,B)=\dist_\C(S_A,S_B)}$.
Without any loss of generality, we may assume $G_1(B)$ has $\lambda_i-\lambda_j<\delta-2\varepsilon$. 
Then
\begin{align*}
    \|(L,R,T)\curvearrowright S_A- S_B\|_F &=\|LS_{A[1]}(R\otimes T)^*-S_{B[1]}\|_F\\
    &\geq \frac{\|LS_{A[1]}S_{A[1]}^* L^*-S_{B[1]}S_{B[1]}^* \|_F}{\|S_A\|_F+\|S_B\|_F} \\
    &\geq \frac{\|LG_1(S_A) L^* -G_1(S_B)\|_F}{\|S_A\|_F+\|S_B\|_F} \\
    &> \frac{\varepsilon }{\|A\|_F+\|B\|_F}
\end{align*}
where the last line holds by Weyl's inequality and invariance of the Frobenius norm under the unitary action.
\end{proof}

In the next lemma, we show that if a column of a unitary matrix $U$ has two sufficiently large entries,
and $D$ is a real diagonal matrix with eigenvalue gap at least $\delta$ then $UDU^*$ will have a proportionally large off-diagonal entry.   

\begin{lemma}\label{lem:orthonormalvectors}
Let $v_1,\dots,v_n$ be 
orthonormal vectors in $\C^n$, and $D$ a diagonal matrix with strictly decreasing positive entries $d_1>d_2>\cdots > d_n>0$,
and $d_i-d_{i+1}>\delta>0$ for $i\in [n-1]$.
For $\varepsilon>0$, if for some $i\neq j$ we have $|v_{ni}|>\varepsilon$ and  $|u_{nj}|>\varepsilon$ then there exists
$k\in [n-1]$ such that $\angles{v_n, Dv_k} > \frac{\varepsilon \delta}{\sqrt{2(n-1)}}$.
\end{lemma}

\begin{proof}
Observe that $\langle v_n, Dv_k\rangle = \langle Dv_n, v_k\rangle$, since $D$ is symmetric. 
Let $w \coloneqq Dv_n-\angles{Dv_n,v_n}v_n$, and observe that $w\in \mathrm{span}\{v_1,\dots,v_{n-1}\}$.
By Parseval's identity, 
\begin{equation}\label{eqn:parseval}
    \sum_{k=1}^{n-1}|\angles{Dv_n,v_k}|^2=\sum_{k=1}^{n-1}|\angles{w,v_k}|^2=\|w\|^2.
\end{equation}
Hence $\max_{k\in [n-1]} |\angles{Dv_n,v_k}|^2\geq \|w\|^2/(n-1)$. 
It remains to show that $\|w\|^2 \geq \varepsilon^2\delta^2/\sqrt{2}$. 
We have 
\[\|w\|^2 = \|Dv_n\|^2 -|\angles{Dv_n,v_n}|^2=\sum_{k=1}^n d_k^2|v_{nk}|^2-\left(\sum_{k=1}^n d_k|v_{nk}|^2\right)^2.\]
Let $X$ be the random variable which takes the value $d_k$ with probability $|v_{nk}|^2$ (this is well
defined since $\|v_n\|=1$). Then \cref{eqn:parseval} gives $\|w\|^2=\mathrm{Var}(X)$.
Let $Y$ be the indicator random variable,
\[Y=\begin{cases}1 &\text{if } X=d_i\text{ or } X=d_j,\\ 0 & \text{otherwise.}\end{cases}\]
By the law of total variance, 
\begin{align*}
    \mathrm{Var}(X)&=\E[\mathrm{Var}(X|Y)] +\mathrm{Var}(\E[X|Y])\\
    &\geq \PP(Y=1)\cdot\mathrm{Var}(X|Y)\\
    &=(|v_{ni}|^2+|v_{nj}|^2) \cdot \frac{|v_{ni}|^2|v_{nj}|^2}{(|v_{ni}|^2+|v_{nj}|^2)^2}(d_i-d_j)^2\\
    &> \frac{\varepsilon^2\delta^2}{2}
\end{align*}
where the last inequality follows from the fact that $|d_i-d_j|>\delta$, and the fact that
$|u_i|^2,|u_j|^2>\varepsilon^2$.
Therefore $\|w\|>\frac{\varepsilon \delta}{\sqrt{2}}$ and the result holds.
\end{proof}

\begin{lemma}\label{lem:diagonalApproximation}
Let $A, B\in \C^{n\times n\times n}$ such that $G_d(A),G_d(B)$ are diagonal matrices with positive real entries sorted in decreasing order with minimum gap $\delta$ for all $d\in [3]$. Suppose further that there exist $L,R,T$ unitary such that $\|(L,R,T)\curvearrowright A- B\|_F<\varepsilon$, for some $0<\varepsilon < \delta/(\|A\|_F+\|B\|_F)$. Then there exist diagonal phase matrices $D_1,D_2,D_3$ such that $\|(D_1,D_2,D_3)\curvearrowright A- B\|_F< \gamma\varepsilon$ for $\gamma=O(n^2(\|A\|_F+\|B\|_F)^2/\delta)$.
\end{lemma}

\begin{proof}
We have \[\|(L,R,T)\curvearrowright A-B\|_F\geq \frac{\|(L A_{[1]} (R\otimes T)^*)(L A_{[1]} (R\otimes T)^*)^* - B_{[1]}B_{[1]}^*\|_F}{\|A\|_F+\|B\|_F}=\frac{\|LG_1(A)L^* - G_1(B)\|_F}{\|A\|_F+\|B\|_F}.\]
By Lemma~\ref{lem:orthonormalvectors}, if there exists a row $i$ of $L$ with 
\[|L_{ij}|>\varepsilon \delta^{-1}\sqrt{2(n-1)} (\|A\|_F+\|B\|_F) 
\text{ and }|L_{ik}|>\varepsilon \delta^{-1} \sqrt{2(n-1)}(\|A\|_F+\|B\|_F), \text{ for } j\neq k,\] 
then there exists $\ell,m\in [n]$, $\ell\neq m$ such that 

\[\|LG_1(A)L^* - G_1(B)\| \geq |(LG_1(A)L^*)_{\ell m}|>\varepsilon (\|A\|_F+\|B\|_F),\] 
a contradiction. Since the rows of $L$ have unit norm, 
this allows us to restrict attention to $L$ such that 
\[\|L-\Pi\|_F< 2n^{3/2}\varepsilon \delta^{-1} \sqrt{2(n-1)} (\|A\|_F+\|B\|_F)\] 
for a phased permutation
matrix $\Pi$. However, by the ordering of the entries of $G_d(A), G_d(B)$, and the eigenvalue
separation, we must have $\Pi$ diagonal. The same holds for $R,T$ by symmetry. Hence,
we can find $D_1,D_2,D_3$ diagonal phase matrices such that 
\[\|(D_1,D_2,D_3)\curvearrowright A- B\|_F< \varepsilon + 6n^{3/2} \varepsilon \delta^{-1}\sqrt{2(n-1)}\|B\|_F(\|A\|_F+\|B\|_F)<\gamma\varepsilon.
\qedhere
\] 
\end{proof}

We now prove \cref{thm:robustOrthoOrbit}.

\begin{lemma}\label{lem:approxOrbitCorrect}
\cref{alg:approximateOrbitDistance} is correct. 
\end{lemma}
\begin{proof}
Let $K\coloneqq\|A\|_F+\|B\|_F$
and suppose that $G_d(A)$ has simple spectrum and $\varepsilon<\delta/4K$. We need to prove the algorithm outputs \textsf{YES} (that is, $\dist_\C(A,B)<\varepsilon$) and 
$(\widetilde{L},\widetilde{R},\widetilde{T})$ if and only if $(\widetilde{L},\widetilde{R},\widetilde{T}) \in \Unitary(n)^3$ with $\| (\widetilde{L},\widetilde{R}, \widetilde{T}) \curvearrowright \widetilde{A}-\widetilde{B}\|_F<\gamma\varepsilon$.
We proceed step by step. 

If the algorithm terminates at \ref{alg:magnitude}, then clearly $\dist_{\CC}(A,B)>\varepsilon$, since $\dist_\C(A,B)\geq\left|\|A\|_F-\|B\|_F\right|>\left|\|\widetilde{A}\|_F-\|\widetilde{B}\|_F\right|-\varepsilon/2$ and the algorithm outputs the correct answer.

If the algorithm terminates at \ref{alg:stepJNF}, we have outputted \textsf{NO} again (i.e.\
$\dist_{\CC}(A,B)>\gamma\varepsilon$). This output is correct because $\dist_{\CC}(A,B)>\delta/4K>\varepsilon$ by Lemma \ref{lem:traceBound} and the precision bounds in \cref{eqn:precisionBound}.

If the algorithm terminates at \ref{alg:HOSVDcore}, then by \cref{lem:diagonalApproximation}, $\dist_\C(A,B) > \gamma\varepsilon$.

If at \ref{alg:feasibilityCheck}, the linear system is \emph{infeasible},
then for any diagonal phase matrices $D_1,D_2,D_3$, 
it must be the case that for some $(i,j,k)$, 
\[|\widetilde{S_A}(i,j,k)|+|\widetilde{S_B}(i,j,k)|> 2\varepsilon n^2\delta^{-1}K\]\[\text{ and }\] 
\[\arg\left(\frac{\widetilde{S_B}(i,j,k)}{(D_1,D_2,D_3)\curvearrowright\widetilde{S_A}(i,j,k)} \right)>
\arccos\left(\frac{|\widetilde{S_A}(i,j,k)|^2+|\widetilde{S_B}(i,j,k)|^2-2\varepsilon^2 n^4\delta^{-2}K^2}{2\;|\widetilde{S_A}(i,j,k)|\;|\widetilde{S_B}(i,j,k)|}\right),\] 
so, by taking the cosine of each side and rearranging terms,
\[
\|(D_1,D_2,D_3)\curvearrowright\widetilde{S_A}-\widetilde{S_B}\|_F > 2\varepsilon n^2\delta^{-1}K
\]

Hence, by Lemma~\ref{lem:diagonalApproximation}, $\dist_{\CC}(A,B)>\varepsilon$. 
Otherwise, the linear system in \ref{alg:feasibilityCheck} is \emph{feasible},
and for the diagonal sign matrices $D_1,D_2,D_3$ constructed we have 
\[\|(D_1,D_2,D_3)\curvearrowright\widetilde{S_A}-\widetilde{S_B}\|_F<2\varepsilon n^{5/2}\delta^{-1}K,\] since
for each $(i,j,k)\in [n]^3$ we have 
$|((D_1,D_2,D_3)\curvearrowright\widetilde{S_A})(i,j,k)-\widetilde{S_B}(i,j,k)|<2\varepsilon n^2\delta^{-1}K$.
Adding back the errors corresponding to the finite precision calculations, 
we preserve $\dist_\C(S_A,S_B)<\gamma\varepsilon$ and so $\dist_\C(A,B)<\gamma\varepsilon$.
Finally, it is easily verified that 
\[\|(\widetilde{L},\widetilde{R},\widetilde{T})\curvearrowright \widetilde{A}-\widetilde{B}\|_F=\|(D_1,D_2,D_3)\curvearrowright\widetilde{S_A}-\widetilde{S_B}\|_F\]
by construction since the Frobenius norm is invariant under orthogonal actions. Hence, $\widetilde{L},\widetilde{R},\widetilde{T}$ witness
the fact that $\dist_\C(A,B)<\gamma\varepsilon$.
\end{proof}

\begin{lemma}\label{lem:approxOrbitEfficient}
\cref{alg:approximateOrbitDistance} runs in time $O(\poly(n,\log(1/\varepsilon)))$.
\end{lemma}
\begin{proof}
One need only observe
that each step runs in time $O(\poly(n,\log(1/\varepsilon)))$. Indeed, 
\ref{alg:stepJNF} runs in $O(\poly(n,\log(1/\varepsilon)))$ by~\cite{dey2023bit}, 
and \ref{alg:feasibilityCheck} runs in time $O(\poly(n,\log(1/\varepsilon)))$,
by solving the linear system using the interior point method (c.f. e.g. \cite{karmarkar1984}). Note that although a direct application of linear programming cannot address the cyclic mod $2\pi$, the prior conditions on the constants in the problem guarantee a unique possible cyclic shift in each linear inequality which can be easily computed. 
The remaining parts of the algorithm run in  $O(\poly(n,\log(1/\varepsilon)))$ trivially.
\end{proof}

\begin{proof}[Proof of \cref{thm:robustOrthoOrbit}]
The result follows immediately from \cref{lem:approxOrbitCorrect} and \cref{lem:approxOrbitEfficient}.
\end{proof}

\begin{proof}[Proof of \cref{thm:robustOrbit}]
By \cref{thm-repulsion-unscaled}, the tensor $A$ satisfies the conditions in the hypothesis of \cref{thm:robustOrthoOrbit} with
high probability and \ref{alg:magnitude} solves the potential issue of $\|B\|_F$ too large.
\end{proof}

\newpage

\section{Eigenvalue repulsion for tall Wishart matrices}\label{sec-repulsion}
As stated in Theorem~\ref{thm:exactUnitIso}, the random model considered in this paper is (after flattening) a random $n \times p$ matrix with i.i.d.\ sub-Gaussian entries with $p = o(n)$. We call such dependence of $p$ on $n$ the ``tall'' regime. In this section we prove a general quantitative eigenvalue repulsion result for such matrices, extending a recent result in \cite{christoffersen2025gaps} which addresses the $p = \Theta(n)$ case.

Recall that our approach, as with \cite{nguyen2015randommatricestailbounds,christoffersen2025gaps}, reduces to proving that the inner product between a left singular vector $v$ and an independent i.i.d.\ sub-Gaussian vector $X$ is typically large. Such a result is called a ``small-ball probability''. There are essentially two properties of $v$ which can prevent this: \begin{enumerate}
    \item $v$ is sparse, and so effectively kills many coefficients from $X$ outright.
    \item $v$ has additive structure, which interacts with the symmetry and independence of the entries of $X$ to produce cancellations with non-negligible probability.
\end{enumerate}

We follow the establish technique in inverse Littlewood--Offord theory of ruling out these two obstructions with high probability in \cref{sec-incompressible,sec-additivity}. Before that, we discuss in \cref{sec-clos-proof-structure} the paper \cite{christoffersen2025gaps}, which is the foundation of our proof strategy, and then in \cref{subsec-prob-prelims} cover some probability preliminaries.

\subsection{A comparison with \cite{christoffersen2025gaps} on the technical aspects}\label{sec-clos-proof-structure}
In this section, we briefly discuss the proof of the repulsion results in \cite{christoffersen2025gaps}, since our strategy will be to follow that proof template, adapting where necessary to the tall case. \cite{christoffersen2025gaps} themselves adopt the typical inverse Littlewood--Offord theory techniques established in earlier work, so first we give an overview of what the novel techniques of \cite{christoffersen2025gaps} are, and then discuss the novelty of our own proof relative to \cite{christoffersen2025gaps}. We emphasize where our proof follows \cite{christoffersen2025gaps} faithfully or diverges.

\paragraph{Remarks on technical contributions of \cite{christoffersen2025gaps}.}
The core difficulty encountered in \cite{christoffersen2025gaps} is the complex dependence structure of the entries of the Wishart matrix. In the setting of Wigner matrices, entries of the matrix are either independent or determine one another, and due to this simplicity it suffices to control the eigenvalue equation $Ax = \lambda x$ directly. This is not the case for Wishart matrices, and the more complicated structure is handled in \cite{christoffersen2025gaps} by observing that the eigenvalue equation $M^T M x = \lambda x$ can be decomposed into the pair of singular value equations $Mx = \sqrt{\lambda}y$ and $M^T y = \sqrt{\lambda} x$. This converts a question about a matrix with complicated entry dependence into a pair of questions about a matrix with simple dependence structure. Having made this decomposition, \cite{christoffersen2025gaps} then apply established inverse Littlewood--Offord techniques to both the right and left singular vectors of $M$, and conclude by arguing that desired properties proven for each of the singular value equations in isolation hold simultaneously. 

To summarize, the major innovation of \cite{christoffersen2025gaps} is to provide techniques to study spectral events for Wishart matrices into a form amenable to established inverse Littlewood--Offord theory through the singular value equations.

\paragraph{Remarks on our technical contributions.} 
Following \cite{christoffersen2025gaps}, we decompose eigenvalue events for Wishart matrices into the intersection of events for the right and left singular vectors. The core difficulty of our setting, which is the tall regime, is the substantial loss of symmetry between the settings of the right and left singular vectors which is enjoyed in \cite{christoffersen2025gaps}. As a result, although we seek similar conclusions regarding the right and left singular vectors of Wishart matrices, a number of the decompositions and arguments used in \cite{christoffersen2025gaps} fail in our setting. However, what we lose in symmetry in the tall regime, we gain in comparatively strong concentration properties of the singular values of $M$, which is now a tall rectangular matrix. This also allows us to address the complex case (see \cref{subsec-repulsion-complex}), which is not covered in \cite{christoffersen2025gaps} \footnote{The complex case is covered by the work of \cite{han2025simplicity}. We comment on the omission of the complex case in \cite{christoffersen2025gaps} in \cref{subsec-repulsion-complex}.}. We comment on specific adaptations throughout this section (see \cref{remark-net-technique,remark-one-sided,remark-right-sing-vecs}).

Our primary contribution to the repulsion literature is to adapt or replace where necessary the arguments of \cite{christoffersen2025gaps} to account for this new asymmetry. For \cref{lem-incompressible-fixed-val-right,lem-incompressible-fixed-val-left} we adopt a different proof strategy, but suspect that the approach of \cite{christoffersen2025gaps} would work. See \cref{remark-net-technique} for the rationale behind this choice. In \cref{lem-lcd-fixed-val} we adopt a similar strategy to \cite{christoffersen2025gaps}, but the asymmetry of our regime becomes especially problematic, and substantial changes are needed. \cref{prop-incompressible-varying-val,prop-lcd-varying-val,thm-obstructions,theorem-repulsion-singular-vals-squared-real} use the main lemmas to rule out the repulsion obstructions and ultimately conclude the repulsion; their arguments follow their analogues in \cite{christoffersen2025gaps}, but more generally typical of inverse LO results. See e.g.\ \cite{nguyen2015randommatricestailbounds} for a representative predecessor to \cite{christoffersen2025gaps}.

\subsection{Preliminaries on concentration of measure}\label{subsec-prob-prelims}
In what follows, we will be managing a large number of constants, some universal and some depending on the law of $\xi$. We will reuse the symbols $c, \hat{c}, C, K$ to denote quantities which may vary from result to result, but we will indicate where constants do and do not depend on $\xi$, and when the values of constants are changing.

In this paper, we typically consider complex random matrices $M \in \C^{n \times p}$ with entries $\zeta_{ij} = \xi_{ij} + i\xi_{ij}'$ where $\xi_{ij}$ and $\xi_{ij}'$ are all i.i.d.\ real valued random variables. Most relevant results cut over from the real to the complex case straightforwardly, and as so we present primarily the real results, and then discuss elements of the required complex generalizations which are not obvious in Section~\ref{subsec-repulsion-complex}. 

In \cite{christoffersen2025gaps} the authors address Wishart matrices with a significantly weaker concentration assumption on coordinate distribution, namely a finite fourth moment. They do this by intersecting an event of interest with an auxiliary structural event such as $\{\| M \| \geq C \sqrt{n} \}$ to enforce those properties which hold generically with high probability for sub-Gaussian matrices. In this paper we choose to examine the cleaner sub-Gaussian case (see \cref{def-subgaussian}), and therefore present results needed from \cite{christoffersen2025gaps} in only their sub-Gaussian forms. We encourage the interested reader to consult the corresponding results in that paper for their most general form. We expect that our repulsion result in the tall regime generalizes to the fourth moment assumption by choosing the appropriate auxiliary events.

Additionally, we work with mean zero, unit variance random variables $\xi$ and generalize by scaling the main results when needed. Now, recall the following:

\begin{definition}[Sub-Gaussian random variable and vector]\label{def-subgaussian}
    A real valued random variable $\xi$ is said to be sub-Gaussian if it satisfies the tail concentration upper bound
    \begin{equation*}
        \Prob[|\xi| > t] \leq 2\exp(-t^2/K^2)
    \end{equation*}
    for all $t > 0$ and some $K > 0$. We call the minimal $K$ such that the concentration inequality holds the sub-Gaussian constant of $\xi$.

    A random vector $X \in \R^n$ is said to be sub-Gaussian if it possesses sub-Gaussian marginals. Additionally, it is called mean-zero if $\E[X] = 0$, and isotropic if $\E[XX^\top] = I$.
\end{definition}

In what follows, implied constants will typically depend on $K$, and often only on $K$. For an $n \times p$ matrix $M$, denote by $s_k(M)$ for $1 \leq k \leq p$ the $k$'th largest singular value. Additionally, let $\MinSingVal(M) \coloneqq s_p(M)$ and $\MaxSingVal(M) \coloneqq s_1(M)$.

The singular values of matrices populated by i.i.d.\ sub-Gaussian random variables are well-concentrated, as the next result shows.

\begin{proposition}[Singular values of sub-Gaussian matrix ({\cite[Theorem 4.6.1]{vershynin2018high}})]\label{prop-sing-val-concentration}
    For $n \geq p$, let $M \in \R^{n \times p}$ be a random matrix with i.i.d., mean-zero, isotropic sub-Gaussian rows. Then for any $u > 0$ the containment
    \begin{equation*}
        [\MinSingVal(M),\MaxSingVal(M)] \subseteq [\sqrt{n} - CK^2(\sqrt{p} + t), \sqrt{n} + CK^2(\sqrt{p} + t)]
    \end{equation*}
    holds with probability at least $1 - 2\exp(t^2)$, and where $C$ is a universal constant. 
\end{proposition}

As mentioned at the beginning of Section \ref{sec-repulsion}, the asymmetry of the left and right singular vectors introduced by the tall regime is the major obstacle in this paper. This is partially offset by Proposition~\ref{prop-sing-val-concentration}, which shows that singular values are concentrated in a band of order $\Theta(\sqrt{p})$ around $\sqrt{n}$ with overwhelming probability. Concretely, when $p = o(n)$ we may choose $t = \sqrt{p}$ in which case by Proposition~\ref{prop-sing-val-concentration} we have
\begin{equation}\label{eq-sing-val-range}
       [\MinSingVal(M),\MaxSingVal(M)] \subseteq [\sqrt{n} - 2CK^2\sqrt{p}, \sqrt{n} + 2CK^2\sqrt{p}]
\end{equation}
with probability at least $1 - 2\exp(-p)$. This will allow us to restrict to a smaller pool of candidate singular values.

\subsection{Singular vectors are not sparse}\label{sec-incompressible}
We now define and rule out the first obstruction to small-ball probabilities of $|v^T X|$ described at the beginning of this section. We require a few definitions.

\begin{definition}[L\'evy concentration function]\label{def-levy-concentration-func}
    The L\'evy concentration function is defined as 
    \begin{equation*}
        \LevyFunc(X, \varepsilon) \coloneqq \sup_{a \in \R^n}\mathbb{P}[\|X - a\| \leq \varepsilon]
    \end{equation*}
    where $X \in \R^n$ is a random vector, $\varepsilon$ is a proximity parameter in $(0, 1)$ and $\| \cdot \|$ is the Euclidean norm.
\end{definition}

The function $\LevyFunc$ upper bounds the concentration of a random vector while obscuring the particular point of concentration. $\LevyFunc$ is finite under mild assumptions on $X$. The following pair of results are standard, but we follow the presentation in \cite{christoffersen2025gaps}.

\begin{lemma}[{\cite[Lemma 2.1]{christoffersen2025gaps}}]\label{lem-levy-concentration-finiteness}
    Let $X$ be a mean zero, variance one sub-Gaussian random variable. Then, for all $\varepsilon \in (0, 1)$ there exists a probability $p = p(K, \varepsilon) < 1$ such that
    \begin{equation*}
        \LevyFunc(X, \varepsilon) \leq p
    \end{equation*}
\end{lemma}

Of course, a sub-Gaussian random variable $\xi$ possesses the moment bounds $\E[|\xi|^q]^{1/q} \leq CK\sqrt{q}$ for all $q \geq 1$ for a universal constant $C$ (see Proposition 2.5.2 of \cite{vershynin2018high}) and so constants in \cite{christoffersen2025gaps} which are declared to depend on $\E[|\xi|^4]$ can be taken to depend upon $K$ instead.

The following lemma, known as the \textit{Tensorization lemma}, allows one to deduce finiteness of $\LevyFunc$ for a random vector $X$ with independent coordinates from finiteness of $\LevyFunc$ for the coordinates themselves. Originating as \cite[Lemma 3.4]{vershynin2014invertibility}, we use the version \cite[Lemma 2.2]{christoffersen2025gaps}.

\begin{lemma}[Tensorization ({\cite[Lemma 2.2]{christoffersen2025gaps}})]\label{lem-tensorization}
    Let $X$ be a random variable taking values in $\R^n$ with independent coordinates. If there exists an $\varepsilon > 0$ and a probability $p \in (0, 1)$ such that
    \begin{equation}
        \LevyFunc(X_i, \varepsilon) \leq p
    \end{equation}
    for all $i \in [n]$ then there exists an $\varepsilon' = \varepsilon'(\varepsilon, p) > 0$ and a probability $p' = p'(\varepsilon, p) \in (0, 1)$ such that 
    \begin{equation*}
        \LevyFunc(X, \varepsilon'
        \sqrt{n}) \leq (p')^n
    \end{equation*}
\end{lemma}

A related tool, which allows one to bound $\LevyFunc$ for $\langle x, X \rangle$ with $x \in S^n$ and $X$ a random vector as in Lemma~\ref{lem-tensorization}, is the following.

\begin{lemma}[{\cite[Lemma 2.3]{christoffersen2025gaps}}]\label{lem-sum-concentration}
    Let $X$ be a random vector with independent, mean-zero, variance-one sub-Gaussian random variables with sub-Gaussian constants uniformly bounded by some $K < \infty$. Then for all $\varepsilon \in (0, 1)$ there exists a probability $p = p(K, \varepsilon) \in (0, 1)$ such that for all $x \in S^{n - 1}$ we have that
    \begin{equation*}
        \LevyFunc(\langle x, X \rangle, \varepsilon) \leq p
    \end{equation*}
\end{lemma}

The first obstruction to control of small ball probabilities considers the fraction of zero/near-zero entries in a vector. 

\begin{definition}[Sparse and compressible vectors]\label{def-compression}
    Given $c_0 \in (0, 1)$, a vector $x \in \R^n$ is called $c_0$-sparse if it is supported on at most $\lfloor c_0 n \rfloor$ coordinates, where the support of $x$ is the set of coordinates with value $\neq 0$.

    Additionally, given $c_1 > 0$, a vector $y \in S^{n - 1}$ is said to be $(c_0, c_1)$-compressible (or simply compressible where the parameters are obvious) if there exists $c_0$-sparse $x\in \R^n$ such that $\|y-x\|\leq c_1$. The set of such $y$ is denoted by $\Comp_n(c_0, c_1)$, and we define the complimentary set $\Incomp_n(c_0, c_1) \coloneqq S^{n - 1} \setminus \Comp(c_0, c_1)$.
\end{definition}

So, a vector is sparse if at most a constant fraction of entries are non-zero, and is compressible if it is close in Euclidean distance to being sparse. The above definition partitions the sphere as $S^{n - 1} = \Comp(c_0, c_1) \cup \Incomp(c_0, c_1)$, and $\Comp(c_0, c_1)$ is precisely the set of vectors representing the first obstruction to small-ball probabilities.

We now prove that compressible vectors are with high probability not approximate singular vectors. The strategy, as in \cite{christoffersen2025gaps}, will be to first establish that a compressible vector cannot be an approximate singular vector for any fixed singular value within the range as in~\cref{eq-sing-val-range}. Then, we will union bound over the range of candidate singular values to complete the desired claim.

By Equation~\ref{eq-sing-val-range} we assume that $\sigma \in [\sqrt{n} - K\sqrt{p}, \sqrt{n} + K\sqrt{p}]$ for some constant $K > 0$ depending on the law of $\xi$.

\begin{lemma}[Anti-concentration of right singular vectors]\label{lem-incompressible-fixed-val-right}
    Fix a mean-zero, variance-one sub-Gaussian distribution $\xi$. Let $M$ be a $n \times p$ matrix with i.i.d.\ entries $m_{ij} \sim \xi$, and with $p = \Theta(n^{\zeta})$ for $\zeta \in (0, 1)$. Additionally, let $\sigma \in [\sqrt{n} - K\sqrt{p}, \sqrt{n} + K\sqrt{p}]$ and $c_0, c_1, c > 0$, and define the event
    \begin{equation}\label{eq-incompressible-fixed-val-1}
        \mathcal{C}_R \coloneqq \bigg\{ \inf_{u \in \Comp_p(c_0, c_1)}\| M^\top M u - \sigma^2 u \| \leq c \sqrt{np} \bigg\}.
    \end{equation}
    Then there exist positive values for $c_0, c_1, c$ such that
    \begin{equation}\label{eq-incompressible-fixed-val-3}
        \mathbb{P}[\mathcal{C}_R] \leq \exp(- \Theta(p)),
    \end{equation}
    with $c_0, c_1, c$ as well as the implied constants depending only on the law of $\xi$.
\end{lemma}
\begin{proof}
    We argue by a union bound over $\mathcal{C}_R$. Specifically, since a compressible vector must be close to some sparse vector, and each sparse vector has support given by some subset of coordinates, we union bound over each admissible subset of support coordinates.  
    
    Let $\Comp_{p, T}(c_0, c_1)$ denote the subset of $\Comp_{p}(c_0, c_1)$ which are $c_1$-close to a $c_0$-sparse unit vector with support contained in the indices $T \subseteq [p]$ with $|T| = \lfloor c_0 p \rfloor$. Allowing $T$ to vary, we have that
    \begin{equation*}
        \Comp_{p}(c_0, c_1) \subseteq \bigcup_{T \subseteq I : \, |T| = \lfloor c_0 p \rfloor}\Comp_{p, T}(c_0, c_1)
    \end{equation*}
    and therefore defining the events 
    \begin{equation*}
        \mathcal{C}_{R, T} \coloneqq \bigg\{ \inf_{u \in \Comp_{p, T}(c_0, c_1)}\| M^\top M u - \sigma^2 u \| \leq c \sqrt{np} \bigg\}
    \end{equation*}
    we have 
    \begin{equation*}
        \mathcal{C}_R \subseteq \bigcup_{T \subseteq I : \, |T| = \lfloor c_0 p \rfloor}\mathcal{C}_{R, T}
    \end{equation*}
    and proceed to bound each $\mathcal{C}_{R, T}$. By rearrangement, we may assume that $T = \{ 1, ..., \lfloor c_0 p \rfloor \}$ and we write the decomposition 
    \begin{equation*}
        M = (X | Y)
    \end{equation*}
    where $X$ is an $n \times \lfloor c_0 p \rfloor$ matrix and $Y$ is an $n \times \lceil (1 - c_0) p \rceil$ matrix. By examining the bottom $p - \lfloor c_0 p \rfloor$ coordinates of the eigenvalue term $M^\top M u_0 - \sigma^2 u_0$ we have that
    \begin{equation*}
        \| M^\top  M u_0 - \sigma^2 u_0 \| \geq \|Y^\top X u_0'\|
    \end{equation*}
    where $u_0$ is a $c_0$-sparse vector in $\Comp_{p, T}(c_0, c_1)$ with support $\subseteq T$, and where $u_0'$ is the vector consisting of the first $\lfloor c_0 p \rfloor$ coordinates of $u_0$. Note that we may write 
    \begin{equation}\label{eq:Ytran-X-u0prime}
        \|Y^\top  X u_0'\| = \|X u_0' \| \cdot \|Y^\top (X u_0'/\|X u_0'\|)\| \geq \MinSingVal(X) \cdot \|Y^\top (X u_0'/\|X u_0'\|)\|
    \end{equation}
    since $u_0'$ is a unit vector. Now, $X$ is a sub-Gaussian rectangular matrix of dimensions $n \times \lfloor c_0 p \rfloor$, and so by Proposition~\ref{prop-sing-val-concentration} we have that
    \begin{equation}\label{eq:min-sing-val-X}
    \MinSingVal(X) \geq \sqrt{n} - K\sqrt{\lfloor c_0 p \rfloor}
    \end{equation}
    with probability at least $1 - e^{-Cn}$ for constants $K, C$ depending only on the law of $\xi$. Now, having controlled $\MinSingVal(X)$ we need to deal with $\|Y^\top (X u_0'/\|X u_0'\|)\|$. Observe that, having conditioned on $X$, that 
    \begin{equation*}
    \|Y^\top (X u_0'/\|X u_0'\|)\| \geq \inf_{x \in \mathrm{span}(X): \, \|x\| = 1}\|Y^\top x\|
    \end{equation*}
    which corresponds to $\MinSingVal(Y^\top U)$ where $U$ is an $n \times \lfloor c_0 p \rfloor$ matrix whose columns form an orthonormal basis of $\mathrm{span}(X)$. Note that $U$ is deterministic conditional upon $X$. Then, conditional on $X$, the rows of $Y^\top U$ are i.i.d sub-Gaussian vectors, and therefore
    \begin{equation}\label{eq:min-sing-val-YtransU}
        \|Y^\top (X u_0'/\|X u_0'\|)\| \geq \inf_{x \in \mathrm{span}(X): \, \|x\| = 1}\|Y^\top x\| = \MinSingVal(Y^\top U) \geq \sqrt{\lceil (1 - c_0)p \rceil} - K \sqrt{\lfloor c_0 p \rfloor})
    \end{equation}
    with (conditional) probability at least $1 - e^{-C'p}$. Plugging Equations~\ref{eq:min-sing-val-X} and~\ref{eq:min-sing-val-YtransU} into Equation~\ref{eq:Ytran-X-u0prime}, and choosing $\hat{c} > 0$ such that $\hat{c}\sqrt{np} \leq (\sqrt{n} - K\sqrt{\lfloor c_0 p \rfloor})(\sqrt{\lceil (1 - c_0)p \rceil} - K \sqrt{\lfloor c_0 p \rfloor}))$ for sufficiently large $n$, we have that 
    \begin{equation*}
        \|Y^\top (X u_0')\| \geq \hat{c}(\sqrt{np})
    \end{equation*}
    with probability at least $(1 - e^{-Cn})(1 - e^{-Cp})$ by the independence of $X$ and $Y$. This yields the statement 
    \begin{equation}
        \inf_{u \in \Comp_{p, T}(c_0, c_1) \cap c_0\mathrm{-sparse}}\|M^\top  M u - \sigma^2 u\| \geq \hat{c}\sqrt{np}
    \end{equation}
    with probability at least $(1 - e^{-Cn})(1 - e^{-Cp})$. By the reverse triangle inequality, we have that for $u \in \Comp_{p, T}(c_0, c_1)$ that 
    \begin{equation*}
        \|M^\top M u - \sigma^2 u\| \geq \|M^\top M u_0 - \sigma^2 u_0\| - \|M^\top M (u - u_0) - \sigma^2 (u - u_0)\| 
    \end{equation*}
    Now note that
    \begin{equation*}
        \|M^\top M (u - u_0) - \sigma^2 (u - u_0)\|  \leq c_1 \|M^\top M - \sigma^2I\| = c_1 \max_{1 \leq i \leq p}|s_i(M)^2 - \sigma^2| \leq c_1 \Theta(\sqrt{np})
    \end{equation*}
    since $\sigma$ lies in a window around $\sqrt{n}$ of width $\Theta(\sqrt{p})$, and so do all the singular values of $M$ with probability at least $1 - e^{-c'p}$. Choosing $c_0, c_1 > 0$ sufficiently small, we have that 
    \begin{equation*}
        \mathbb{P}[\mathcal{C}_{R,T}] \leq O(e^{-Cn}) + O(e^{-Cp}) \leq O(e^{-Cp})
    \end{equation*}
    with universal implied constants in the above display. Now, by Stirling's approximation we have that 
    \begin{equation*}
        \binom{p}{c_0 p} \leq e^{H(c_0)p}
    \end{equation*}
    where
    \begin{equation*}
        H(x) \coloneqq -x\log x - (1 - x)\log (1 - x)
    \end{equation*}
    is the binary entropy function (see Theorem 3.1 of \cite{galvin2014three}). A union bound leads to
    \begin{equation*}
        \mathbb{P}[\mathcal{C}_R] \leq \sum_{T \subseteq [p]: \, |T| = c_0 p}\mathbb{P}[\mathcal{C}_{R, T}] \leq 2e^{H(c_0)p - Cp}
    \end{equation*}
     with $c = \hat{c} - c_1$. Note that since $\hat{c}$ does not depend on $c_1$, and since decreasing $c_1$ only improves our concentration, we can choose $c_1$ sufficiently small that $c > 0$. This concludes the argument.
\end{proof}

We now establish a similar statement for the left singular vectors. Unlike with the right singular vector case, in which the first matrix that the singular vector hits has the favorable dimensions, here the second matrix has the favorable properties and so we require a more careful lower bound.

\begin{lemma}[Anti-concentration of left singular]\label{lem-incompressible-fixed-val-left}
    Fix a mean-zero, variance-one sub-Gaussian distribution $\xi$. Let $M$ be a $n \times p$ matrix with i.i.d.\ entries $m_{ij} \sim \xi$, and with $p = \Theta(n^{\zeta})$ for $\zeta \in (0, 1)$. Additionally, let $\sigma \in [\sqrt{n} - K\sqrt{p}, \sqrt{n} + K\sqrt{p}]$ and $c_0, c_1, c > 0$, and define the event
    \begin{equation}\label{eq-incompressible-fixed-val-1}
        \mathcal{C}_L \coloneqq \bigg\{ \inf_{u \in \Comp_n(c_0, c_1)}\|M M^\top u - \sigma^2 u\| \leq c n \bigg\}.
    \end{equation}
    Then there exist positive values for $c_0, c_1, c$ such that
    \begin{equation}\label{eq-incompressible-fixed-val-3}
        \mathbb{P}[\mathcal{C}_L] \leq \exp(- \Theta(n)),
    \end{equation}
    with $c_0, c_1, c$ as well as the implied constants depending only on the law of $\xi$.
\end{lemma}
\begin{proof}
   Similar to the treatment of right singular vectors, let $\Comp_{n, T}(c_0, c_1)$ denote the subset of $\Comp_{n}(c_0, c_1)$ which are $c_1$-close to a $c_0$-sparse unit vector with support contained in the indices $T \subseteq [n]$ with $|T| = \lfloor c_0 n \rfloor$. Allowing $T$ to vary, we have that
    \begin{equation*}
        \Comp_{n}(c_0, c_1) \subseteq \bigcup_{T \subseteq I : \, |T| = \lfloor c_0 n \rfloor}\Comp_{n, T}(c_0, c_1)
    \end{equation*}
    and therefore defining the events 
    \begin{equation*}
        \mathcal{C}_{L, T} \coloneqq \bigg\{ \inf_{v \in \Comp_{n, T}(c_0, c_1)}\|MM^\top v - \sigma^2 v\| \leq c n \bigg\}
    \end{equation*}
    we have that 
    \begin{equation*}
        \mathcal{C}_L \subseteq \bigcup_{T \subseteq I : \, |T| = \lfloor c_0 n \rfloor}\mathcal{C}_{L, T}
    \end{equation*}
    and may proceed to bound each $\mathcal{C}_{L, T}$. By rearrangement, we may assume that $T = \{ 1, ..., \lfloor c_0 n \rfloor \}$ and we write the decomposition 
    \begin{equation*}
        M = \bigg(\frac{X}{Y}\bigg)
    \end{equation*}
    where $X$ is a $\lfloor c_0 n \rfloor \times p$ matrix and $Y$ is a $\lceil(1 - c_0)n \rceil \times p$ matrix. We have that
    \begin{equation}\label{eq-left-compressible-decomp}
        \| MM^\top v_0 - \sigma^2 v_0\| = \|YX^\top v_0' + XX^\top v_0' - \sigma^2 v_0'\| = \|YX^\top v_0'\| + \|XX^\top v_0' - \sigma^2 v_0'\|
    \end{equation}
    where $v_0$ is a $c_0$-sparse vector in $\Comp_{n, T}(c_0, c_1)$ with support $\subseteq T$, and where $v_0'$ is the vector consisting of the first $\lfloor c_0 n \rfloor$ coordinates of $v_0$. 
    
    Conditional on $X$, let $v_{\perp}$ be the projection of $v_0'$ onto the kernel of $X^\top $, and define $v_{\|} \coloneqq v_0' - v_{\perp}$. Equation~\ref{eq-left-compressible-decomp} when squared admits the lower bound
    \begin{align*}
        \| MM^\top v_0 - \sigma^2 v_0\|^2 &\geq \|YX^\top v_{\|}\|^2 + \|XX^\top (v_{\perp} + v_{\|}) - \sigma^2 (v_{\perp} + v_{\|})\|^2 \\
        &= \|YX^\top v_{\|}\|^2 + \|(XX^\top v_{\|} - \sigma^2 v_{\|})  - \sigma^2 v_{\perp}\|^2 \\
        &= \|YX^\top v_{\|}\|^2 + \|(XX^\top v_{\|} - \sigma^2 v_{\|})\|^2 + \|\sigma^2 v_{\perp}\|^2 \\
        &\geq \|YX^\top v_{\|}\|^2 + \|\sigma^2 v_{\perp}\|^2
    \end{align*}
    by an application of orthogonality at the second equality. The $v_{\perp}$ term will allow us to accommodate for the possibility that $v_{\|}$ is small. Now, since the row space of $X^\top $ is spanned by the singular vectors of $X^\top $ corresponding to only the non-zero singular values of $X^\top $, following the technique in Lemma~\ref{lem-incompressible-fixed-val-right} we have that
    \begin{equation*}
        \|YX^\top v_{\|}\| \geq \|v_{\|}\| \cdot \MinSingVal(X) \cdot \MinSingVal(Y)
    \end{equation*}
    and so combining we have that
    \begin{equation*}
        \| MM^\top v_0 - \sigma^2 v_0\|^2 \geq \|v_{\|}\|^2 \cdot (\MinSingVal(X)^2 \cdot \MinSingVal(Y)^2) + \sigma^4 \|v_{\perp}\|^2
    \end{equation*}
    since $v_0'$ is a unit vector, the lower bound is with probability at least $(1 - e^{-Cn})^2$ the convex combination of terms which are $\Theta(n^2)$. It follows that there exists some $\hat{c} > 0$ (independent of $c_1$) such that $\hat{c}^2 n^2 \leq \min(\MinSingVal(X)^2 \cdot \MinSingVal(Y)^2)$ with probability at least $(1 - e^{-Cn})^2$.
    
    . As $v_0'$ was arbitrary within $\Comp_{n, T}(c_0, c_1) \cap c_0\mathrm{-sparse}$, we extend to the entirety of $\Comp_{n, T}(c_0, c_1)$ by observing that uniformly over the cluster with probability at least $(1 - e^{-Cn})^2$ that
    \begin{equation*}
        \|MM^\top  u - \sigma^2 u\| \geq \|MM^\top  u_0 - \sigma^2 u_0\| - \|MM^\top  (u - u_0) - \sigma^2 (u - u_0)\| \geq \hat{c} n - c_1 \Theta(n)
    \end{equation*}
    again by the upper bounds on $\sigma^2$ and $\|M\|$. By a union bound over the clusters $\Comp_{n, T}(c_0, c_1)$ and an application of Stirling's approximation we have that
    \begin{equation*}
         \mathbb{P}[\mathcal{C}_L] \leq 2e^{H(c_0)n - Cn}
    \end{equation*}
    where again we choose $c_0, c_1$ sufficiently small and set $c = \hat{c} - c_1$, noting that $c > 0$ since $\hat{c}$ is independent of $c_1$ and $c_1$ can be made arbitrarily small.
\end{proof}

An important distinction between Lemmas \ref{lem-incompressible-fixed-val-right} and \ref{lem-incompressible-fixed-val-left} is that it was only necessary to use the sharper control on the concentration of the singular values of $M$ in the former and not the latter. 

Note in Lemmas~\ref{lem-incompressible-fixed-val-right} and \ref{lem-incompressible-fixed-val-left} reducing the quantities $c_0, c_1, c$ only improves the concentration probability and we may choose the relevant constants for both the right and left singular vectors to be identical. Indeed, the proofs of \cref{lem-incompressible-fixed-val-right,lem-incompressible-fixed-val-left} establish existence of the appropriate constants by choosing $c$ which depends on $c_1$, one may subsequently lower $c$ by monotonicity of measure, and $c_0, c_1$ by monotonicity of set containment.

We can now rule out the first small-ball obstruction with overwhelming probability.

\begin{proposition}[Compressible singular vectors are unlikely]\label{prop-incompressible-varying-val}
    Let $\mathcal{G}$ be the event that there exists either a vector $u \in \Comp_{p}(c_0, c_1)$ and value $\sigma^2 > 0$ such that $M^\top Mu = \sigma^2 u$, and either $u \in \Comp_{p}(c_0, c_1)$ or $Mu = v \in \Comp_{p}(c_0, c_1)$. Then 
    \begin{equation*}
        \mathbb{P}[\mathcal{G}] \leq \exp(-\Theta(p))
    \end{equation*}
    for sufficiently large $n$ and with implied constants depending only on the law of $\xi$.
\end{proposition}
\begin{proof}
    The strategy is to choose some discretization of the set of candidate singular values, treat each fixed point in this discretization via \cref{lem-incompressible-fixed-val-right,lem-incompressible-fixed-val-left}, and conclude with a union bound. So, this is a net argument over the singular values.

    Let $\mathcal{G}_R$ denote that there exist a compressible right singular vector such as in the Proposition statement, and let $\mathcal{G}_L$ denote the left singular vector event, such that $\mathcal{G} = \mathcal{G}_R \cup \mathcal{G}_L$. We deal with the right singular vectors as the left case follows identically.
    
    Suppose that such a $u \in \Comp_{p}(c_0, c_1)$ and $\sigma^2$ exist, and let $\sigma_0^2$ be an integer multiple of $c \sqrt{np}$ (with $c$ as in Lemma~\ref{lem-incompressible-fixed-val-right}) in the range $[(\sqrt{n} - K\sqrt{p})^2, (\sqrt{n} + K\sqrt{p})^2]$. Then, there are $O(2K\sqrt{np})/c\sqrt{np}$ many such values, and $|\sigma^2 - \sigma_0^2| \leq c \sqrt{np}$, so we have 
    \begin{equation*}
        \|M^\top  M u - \sigma_0^2 u \| = |\sigma^2 - \sigma_0^2| \leq c \sqrt{np}
    \end{equation*}
    which occurs with probability at most $\exp(- \Theta(p))$ by Lemma~\ref{lem-incompressible-fixed-val-right}. By a union bound over the polynomially many values of $\sigma_0^2$ we have that 
    \begin{equation*}
        \mathbb{P}[\mathcal{G}_R] \leq \exp(-\Theta(p))
    \end{equation*}
    for sufficiently large $n$ with implied constants depending only on the law of $\xi$. By a similar argument and using Lemma~\ref{lem-incompressible-fixed-val-left} we determine that
    \begin{equation*}
        \mathbb{P}[\mathcal{G}_L] \leq \exp(-\Theta(n))
    \end{equation*}
    again with implied constants depending only on the law of $\xi$. A final union bound then yields the claim.
\end{proof}

We conclude this section with a few remarks.

\begin{remark}[Superficiality of upper bounds and grid granularity]
    Note that the particular upper bounds on the infima of Lemmas~\ref{lem-incompressible-fixed-val-right} and \ref{lem-incompressible-fixed-val-left} are mostly superficial. Once one can prove that the relevant infimum is $\leq n^{\alpha}$ for some $\alpha \in \R$ with probability at least $1 - \exp(-\Theta(p))$, the impact of $\alpha$ is exclusively on the size of the singular value grid, which contributes an $n^{-\xi'}$ multiplicative factor via a union bound which is ultimately dominated by the exponentially small probability of the fixed $\sigma$ event. The particular thresholds are merely the result of technique limitations.
\end{remark}

\begin{remark}[Net technique]\label{remark-net-technique}
    Note that in \cite{christoffersen2025gaps} the authors prove their analogues of Lemmas~\ref{lem-incompressible-fixed-val-right} and \ref{lem-incompressible-fixed-val-left} using a net over the compressible vectors. We suspect that such an argument, when supplemented with the additional singular value information of Equation~\ref{eq-sing-val-range} as well as the more careful lower bound in Lemma~\ref{lem-incompressible-fixed-val-right}, would be sufficient to produce Proposition~\ref{prop-incompressible-varying-val}. 
    
    In this paper we choose to present a support union argument, as opposed to the net union argument, as it is a little more simple, yields a slightly more refined result with stronger intermediate claims, a better dependence of the rate of exponential convergence on $c_0$ and an independence of the rate on $c_1$. In particular, the support union bound allows one to consider the setting in which $c_0$ and $c_1$ vary, which is obstructed by the net union multiplicative term $(9/c_0 c_1)^{c_0 n}$.
\end{remark}

\subsection{Left singular vectors do not have strong additive structure}\label{sec-additivity}

The second obstruction to small-ball probabilities, which is the existence of additive structure, is quantified by a concept known as \textit{least common denominator} (LCD), introduced in \cite{rudelson2008littlewood}. Following \cite{christoffersen2025gaps}, we use a variant known as \textit{regularized} LCD.

\begin{definition}[LCD and regularized LCD]\label{def-lcd}
    Let $\kappa > 0$ and $\gamma \in (0, 1)$. For a unit vector $x \in S^{n - 1}$ we define $\LCD_{\kappa, \gamma}(x)$ as 
    \begin{equation*}
        \LCD_{\kappa, \gamma}(x) \coloneqq \inf \{ \theta > 0: \, \dist(\theta x, \Z^{n}) < \min(\gamma \|\theta x \|, \kappa) \}
    \end{equation*}
    Additionally, we define the regularized LCD for $\alpha \in (0, c')$ with $c' = c_0 c_1^2 /4$, $0 < \alpha < c'/4$ and $x \in \Incomp_n(c_0, c_1)$, which we denote by $\rLCD_{\kappa, \gamma}(x, \alpha)$, as
    \begin{equation*}
        \rLCD_{\kappa, \gamma}(x, \alpha) \coloneqq \sup \{ \LCD_{\kappa, \gamma}(x_I/\|x_I\|): I \subseteq \spread(x), \, |I| =  \lceil \alpha n \rceil \}
    \end{equation*}
    where 
    \begin{equation*}
        \spread(x) \coloneqq \{i \in [n]: \, \frac{c_1}{\sqrt{2n}} \leq |x_i| \leq \frac{1}{\sqrt{c_0 n}} \}
    \end{equation*}
\end{definition}

$\LCD_{\kappa, \gamma}(x)$ finds the smallest $\theta > 0$ such that $\theta x$ is approximately integer valued, which describes how close the entries of $x$ are to being approximately embedded in an arithmetic progression. The value $\kappa$ configures the notion of ``approximate'', and $\gamma \|\theta x \|$ prevents the approximation of $\theta x$ by an integer vector direction which deviates significantly from $x$. $\rLCD_{\kappa, \gamma}(x, \alpha)$ is a more robust version of $\LCD_{\kappa, \gamma}(x)$; since $x \in S^{n - 1}$, $\LCD_{\kappa, \gamma}(x)$ can be large by virtue of the entries of $x$ being ``spread-out'' rather than having genuine arithmetic structure, so $\rLCD_{\kappa, \gamma}(x, \alpha)$ considers the arithmetic structure over all such well-spread subsets of $x$. See \cite{rudelson2010lecture} for a more thorough discussion.

The following describes the quantitative impact of regularized LCD on concentration.

\begin{lemma}[{\cite[Lemma 4.1]{christoffersen2025gaps}}]\label{lem-lcd-helper-concentration}
    Let $X$ be a random vector in $\R^n$ with i.i.d.\ coordinates, and assume that there exist $\varepsilon_0, p_0 > 0$ such that $\LevyFunc(X_i, \varepsilon_0) \leq 1 - p_0$ for all $i \in [n]$. If $x \in \Incomp_n(c_0, c_1)$ and $\kappa > 0$, if $\varepsilon$ satisfies 
    \begin{equation*}
        \varepsilon \geq \frac{\sqrt{\alpha}}{c_0 \rLCD_{\kappa, \gamma}(x, \alpha)}
    \end{equation*}
    then 
    \begin{equation}
        \LevyFunc(\langle X, x \rangle, \varepsilon) \leq C \bigg( \frac{\varepsilon}{\gamma c_1 \sqrt{\alpha}} + e^{-\Theta(\kappa^2)} \bigg)
    \end{equation}
    where the implied constants depend only on the law of $X_1$.
\end{lemma}

Note that the constants above in our setting only pertain to the individual random variables which we populate our matrix $M$ with, and so are fixed.

 Since small $\rLCD$ represents strong additive structure, we aim to prove a high probability lower bound on the $\rLCD$ of a left singular vector. Recall that for a metric space $(X, d)$ that an $\varepsilon$-net of a subset $S \subseteq X$ is a subset $\mathcal{N} \subseteq S$ such that for all $s \in S$ that $d(\mathcal{N}, s) < \varepsilon$. Nets will be useful structures with which to construct union bound arguments, and the following result bounds the size of an $\varepsilon$-net of the sub-level sets with respect to regularized LCD.

\begin{lemma}[{\cite[Lemma B.1]{nguyen2015randommatricestailbounds}}]\label{lem-lcd-helper-level-size}
    Consider $D \geq 1$ and $n^{-c} \leq \alpha \leq c'/4$. We denote by $S_D$ the sub-level set
    \begin{equation*}
        S_D \coloneqq \bigg\{x \in \Incomp_n(c_0, c_1): \, \rLCD_{\kappa, \gamma}(x, \alpha) \leq D \bigg\}
    \end{equation*}
    $S_D$ has a $\frac{\kappa}{D\sqrt{\alpha}}$-net $\net$ with 
    \begin{equation*}
        |\net| \leq \frac{C^nD^{n + 2/\alpha}}{(\alpha n)^{c'n/4}}
    \end{equation*}
    where $C$ is a universal constant.
\end{lemma}

A more refined version of the above result, which considers a slice of the sub-level set, enables tighter control on a union bound.

\begin{lemma}[{\cite[Definition 4.6]{christoffersen2025gaps}}]\label{lem-lcd-helper-level-size-refined}
    In analogy to Lemma~\ref{lem-lcd-helper-level-size}, we define the level set slice
    \begin{equation*}
        S_{D, k} \coloneqq \bigg\{x \in \Incomp_n(c_0, c_1): \, 2^{-(k + 1)}D \leq \rLCD_{\kappa, \gamma}(x, \alpha) \leq 2^{-k}D \bigg\}
    \end{equation*}
    where $k \in [0, \log_2 D]$, which has a $\frac{2^k\kappa}{D\sqrt{\alpha}}$-net $\net$ with 
    \begin{equation*}
        |\net| \leq \frac{C^nD^{n + 2/\alpha}}{(\alpha n)^{c'n/4}} 2^{-kn - (2k/\alpha)}
    \end{equation*}
    where $C$ is an universal constant.
\end{lemma}

Now, we rule out left singular vectors with low LCD. In \cite{christoffersen2025gaps}, the authors' primary focus is on the singular values of a general family of matrices $\Sigma^{1/2}M$ (with $\Sigma$ positive definite), and their repulsion is achieved by using high LCD bounds on left singular vectors and incompressibility of right singular vectors. However, due to the symmetry of the $p = \Theta(n)$ regime the authors can and do prove high LCD bounds on the right singular vectors too, yielding analogous repulsion statements for the matrices $M \Sigma^{1/2}$. In our setting it is neither necessary nor possible to establish with our current technique high LCD bounds of the right singular vectors of $M$, and we will comment on this obstruction after Lemma~\ref{lem-lcd-fixed-val}. 

First, we require a result from \cite{christoffersen2025gaps}; the following statement can be extracted from the proof of Lemma 4.4.

\begin{lemma}[{\cite[Half of Lemma 4.4]{christoffersen2025gaps}}]\label{lem-approx-singular}
    Fix a mean-zero, variance-one sub-Gaussian distribution $\xi$. Let $M$ be a $n \times p$ matrix with i.i.d.\ entries $m_{ij} \sim \xi$, and with $p = \Theta(n^{\zeta})$ for $\zeta \in (0, 1)$. Recall that $c' = c_0 c_1^2/4$. Then there exists a constant $C > 0$ depending only on the law of $\xi$ such that for any $x \in \Incomp_p(c_{0}, c_{1})$ and $y \in \Incomp_n(c_{0}, c_{1})$, $0 < \alpha < c'/4$, and $\varepsilon > 0$ which satisfies
    \begin{equation*}
        \varepsilon \geq \frac{\sqrt{\alpha}}{c_0 \rLCD_{\kappa, \gamma}(x, \alpha)}
    \end{equation*}
    we then have
    \begin{equation*}
        \sup_{a \in \R^n}\mathbb{P}[ \{\|Mx - a\| < \varepsilon \sqrt{n}\}] \leq C\bigg(\frac{\varepsilon\sqrt{2}}{\gamma c_{1} \sqrt{\alpha}} + e^{-\Theta(\kappa^2)} \bigg)^{n - \lfloor \alpha n \rfloor}
    \end{equation*}
\end{lemma}
We can now rule out low LCD for approximate left singular vectors with a fixed target singular value for an analogous parameter range to \cite{christoffersen2025gaps}.

\begin{lemma}\label{lem-lcd-fixed-val}
   Let $M$ satisfy the assumptions of \cref{lem-approx-singular}. There exists some $c > 0$ depending only on $c_0$ and $c_1$ such that for any fixed $\sigma_0^2 \in [(\sqrt{n} - K\sqrt{p})^2, (\sqrt{n} + K\sqrt{p})^2]$ and any fixed $a \in \R^n$, under the restrictions
    \begin{equation*}
        \kappa = n^{c}, \, D \leq (\kappa)^{1/\alpha}, \, \alpha \geq 1/\sqrt{\kappa}
    \end{equation*}
    we have that with
    \begin{equation*}
        \beta = \frac{\kappa}{\sqrt{\alpha}D}
    \end{equation*}
    that the event
    \begin{equation*}
        \event \coloneqq \{ \exists (u, v) \in \Incomp(c_{0}, c_{1}) \times S_{D}: \|Mu - \sigma_0 v - a \| \leq \beta \sqrt{n} \} 
    \end{equation*}
    has probability at most $ \exp(-O(n \log n))$ with implied constants depending only on the law of $\xi$, and where $S_{D}$ denotes the $\rLCD$ sub-level sets.
\end{lemma}
\begin{proof}
    We can decompose $\event$ into the events $\event_{k}$ which signify that $\event$ occurs with $v \in S_{D, k}$ as in Lemma~\ref{lem-lcd-helper-level-size-refined}. We further decompose $\event_{k}$ into the events $\event_{k, -}$ and $\event_{k, +}$, which denote that $\event$ holds with $\rLCD_{\kappa, \gamma}(v, \alpha) \leq \rLCD_{\kappa, \gamma}(u, \alpha)$ and $\rLCD_{\kappa, \gamma}(v, \alpha) \geq \rLCD_{\kappa, \gamma}(u, \alpha)$ respectively. Finally, we decompose $\event_{k, +}$ into the events $\event_{k, +, l}$ for $1 \leq l \leq k$ which denote that $u \in S_{D, l}$. This decomposition can be summarized as 
    \begin{equation*}
        \event \subseteq \bigcup_{1 \leq k \leq \log D} \bigg(\event_{k, -} \cup \bigg(\bigcup_{1 \leq l \leq k}\event_{k, +, l} \bigg) \bigg)
    \end{equation*}
    and we bound each individual event. First, consider $\event_{k, -}$. Then $\rLCD_{\kappa, \gamma}(x, \alpha) \geq D2^{-(k + 1)}$, and we select a $(\beta_{k})$-net of the left level set slice called $\mathcal{L}_k$ and a $(\beta_{k})$-net of the super-level set called $\mathcal{R}_k$ which satisfy the cardinality bounds
    \begin{equation*}
        |\mathcal{L}_k| \leq \frac{(CD)^n}{(2^k)^n(\sqrt{\alpha n})^{c'n/2}}(2^{-k}D)^{2/\alpha} ,\,
        |\mathcal{R}_k| \leq \bigg(\frac{C}{\beta_{k}}\bigg)^p
    \end{equation*}
    by the standard sphere bounds and Lemma~\ref{lem-lcd-helper-level-size-refined}. Additionally, define
    \begin{equation*}
        \beta_{k} = \frac{2^k\kappa}{\sqrt{\alpha}D}
    \end{equation*}
    We can assume without loss of generality that the elements of the left and right nets are in the sub-level slice set and sub-level set respectively. Now, observe that on the event $\event_{R, k}$ that for some $u_0, v_0 \in \mathcal{R}_{k} \times \mathcal{L}_{k}$ that
    \begin{align}\label{eq-right-approximate-fixed}
        \|Mu_0 - \sigma_0 v_0 - a\| &\leq \|M(u - u_0)\| + \sigma_0\|v - v_0\| + \|Mu - \sigma_0 v - a\| \\
        &\leq \|M\| \beta_{k} + \sqrt{n}\beta_{k} + \beta_{R}\sqrt{n} = O(\beta_{k}\sqrt{n})
    \end{align}
    and that
    \begin{equation*}
       \varepsilon = \beta_{k} = \frac{2^k\kappa}{\sqrt{\alpha}D} \geq \frac{\sqrt{\alpha}}{\rLCD_{\kappa, \gamma}(u_0, \alpha)}
    \end{equation*}
    Together these facts render the union bound feasible. Applying Lemma~\ref{lem-approx-singular} we have that
    \begin{align*}
        \mathbb{P}[\event_{L, k, -}] &\leq  C_5 \bigg(\frac{C}{\beta_{k}}\bigg)^p  \\
        &\times \bigg(\frac{(CD)^n}{(2^k)^n(\sqrt{\alpha n})^{c'n/2}}(2^{-k}D)^{2/\alpha}\bigg) \\
        &\times \bigg(\frac{\varepsilon\sqrt{2}}{\gamma c_{1} \sqrt{\alpha}} + e^{-\Theta(\kappa^2)} \bigg)^{n - \lfloor \alpha n \rfloor}
    \end{align*}
    Now, note that
    \begin{align*}
        \bigg(\textcircled{1} \bigg) &=  \bigg(\frac{C}{\beta_{k}}\bigg)^p = \exp(\pm O(p \log n)) \\
        \bigg(\textcircled{2} \bigg) &=  \bigg(\frac{(CD)^n}{(2^k)^n(\sqrt{\alpha n})^{c'n/2}}(2^{-k}D)^{2/\alpha}\bigg) = \exp(\pm O(n \log n)) \\
        \bigg(\textcircled{3} \bigg) &= \bigg(\frac{\varepsilon\sqrt{2}}{\gamma c_{1} \sqrt{\alpha}} + e^{-\Theta(\kappa^2)} \bigg)^{n - \lfloor \alpha n \rfloor} = \exp(\pm O(n \log n))
    \end{align*}
    and therefore we can safely ignore term $\textcircled{1}$ and address terms $\textcircled{2}$ and $\textcircled{3}$. We have, for sufficiently large $n$, that
    \begin{align*}
        \bigg(\textcircled{2} \bigg) \cdot \bigg(\textcircled{3} \bigg) &\leq \bigg(\frac{(CD)^n}{(2^k)^n(\sqrt{\alpha n})^{c'n/2}}(2^{-k}D)^{2/\alpha}\bigg) \cdot \bigg(\frac{2^k \kappa \sqrt{2}}{\gamma c_{1} \alpha D} + e^{-\Theta(\kappa^2)} \bigg)^{n - \lfloor \alpha n \rfloor} \\
        &= \frac{2^{(n - \lfloor \alpha n \rfloor)/2}C^n D^{\lfloor \alpha n \rfloor  + 2/\alpha}\kappa^{n - \lfloor \alpha n \rfloor}}{(2^k)^{\lfloor \alpha n \rfloor + 2/\alpha}\alpha^{c'n/4 + n - \lfloor \alpha \rfloor}n^{c'n/4}(c_1\gamma)^{n - \lfloor \alpha n \rfloor}} \\
        &= \frac{2^{(n - \lfloor \alpha n \rfloor)/2}C^n}{(2^k)^{\lfloor \alpha n \rfloor + 2/\alpha}\alpha^{c'n/4 + n - \lfloor \alpha n \rfloor}(c_1\gamma)^{n - \lfloor \alpha n \rfloor}}  \cdot \frac{D^{\lfloor \alpha n \rfloor  + 2/\alpha}\kappa^{n - \lfloor \alpha n \rfloor}}{n^{c'n/4}} \\
        &\leq \frac{2^{(n - \lfloor \alpha n \rfloor)/2}C^n}{(2^k)^{\lfloor \alpha n \rfloor + 2/\alpha}(c_1\gamma)^{n - \lfloor \alpha n \rfloor}}  \cdot \frac{D^{\lfloor \alpha n \rfloor  + 2/\alpha}\kappa^{n - \lfloor \alpha n \rfloor + (c'n/4 + n - \lfloor \alpha n \rfloor)/2}}{n^{c'n/4}} \\
        &\leq \frac{2^{(n - \lfloor \alpha n \rfloor)/2}C^n}{(2^k)^{\lfloor \alpha n \rfloor + 2/\alpha}(c_1\gamma)^{n - \lfloor \alpha n \rfloor}}  \cdot \frac{n^{c(n + n^c + n - \lfloor \alpha n \rfloor + (c'n/4 + n - \lfloor \alpha n \rfloor)/2)}}{n^{c'n/4}} \\
        &= \bigg(\textcircled{a} \bigg) \cdot \bigg(\textcircled{b} \bigg)
    \end{align*}
    Examining \textcircled{b}, we find that by choosing $c$ sufficiently small that $(\textcircled{2}) \cdot (\textcircled{3}) = \exp(-O(n \log n))$ since \textcircled{a} possesses an upper bound which is $\exp(O(n))$ with implied constants independent of $k$. In particular, we require that $c < 1$ to ensure that the numerator is $\exp(O(n \log n))$.

    The above argument proves that $\mathbb{P}[\event_{L, k, +}] \leq \exp(-O(n \log n))$ with implied constants depending only on $\xi$, and in particular independent of $k$. To deal with the events $\event_{L, k, +, l}$, we repeat the above argument with a net parameter of $(\beta_l)$, and refined $\rLCD$ slice nets on both sides. This only improves the strength of the the upper bound, and thus we achieve an $\exp(-O(n \log n))$ upper bound with implied constants independent of $l$ and $k$. 

    Now, there are $\log D$ many events $\event_{L, k, -}$, and $\log D(\log D + 1)/2$ many events $\event_{L, k, +}$, and so a union bound produces an upper bound on $\event_L$ of order $\exp(-O(n\log n))$ and thus completes the proof.
\end{proof}

A few remarks are in order regarding the proof of \cref{lem-approx-singular}. First, if one compares \cref{lem-approx-singular} with Proposition 4.8 of \cite{christoffersen2025gaps}, it becomes clear that our proof is not symmetric in the left and right singular values, but is also not totally decoupled. \cref{remark-two-sided} discusses why we still need to deal to some extent with right singular vectors, and \cref{remark-one-sided} describes why the full symmetry of Proposition 4.8 of \cite{christoffersen2025gaps} is not required here either. Secondly, \cref{remark-right-sing-vecs} describes why our current approach is unable to prove strong LCD lower bounds on the right singular vectors.

\begin{remark}[Right vector involvement in \cref{lem-lcd-fixed-val}]\label{remark-two-sided}
    Given that by assumption the candidate vectors involved in the statement of \cref{lem-lcd-fixed-val} are incompressible, a uniform structure over which to partition and union bound $\Incomp_n$ is not obvious, and we proceed by a net argument. Now, any one sided net for the left-singular vector statement should naturally involve a net with exponent $\Theta(n)$, but by combining \cref{lem-lcd-helper-concentration} with \cref{lem-tensorization} expect to yield a probability with exponent $\Theta(p)$. Such an argument appears insufficient, and the involvement of the right singular vectors is done to introduce a probability term of the necessary order.
\end{remark}

\begin{remark}[Asymmetry of Lemma~\ref{lem-lcd-fixed-val}]\label{remark-one-sided}
    Note that in \cref{lem-lcd-fixed-val}, the upper bound consists of three terms: a left net, a right net, and one probabilistic term. This proof could have very well been executed in the manner of Proposition 4.8 of \cite{christoffersen2025gaps} by using the two-pronged version of \cref{lem-approx-singular} and producing a left-sided analogue of \cref{eq-right-approximate-fixed}. However, the exponent of the second probabilistic term is $\Theta(p)$ and so the entire term is asymptotocally dominated by term $(\textcircled{2}) \cdot (\textcircled{3})$, therefore becoming irrelevant. Moreover, one may verify that the left-sided analogue of \cref{eq-right-approximate-fixed} would yield an $\varepsilon = \Theta(n/p)$ in the left probability, so the probability bound would be diverging as $n \to \infty$, and would actively harm the strength of the upper bound although being asymptotically benign. In \cite{christoffersen2025gaps} all four terms in the bound are of the same order, and both probability terms are helpful, motivating their two-pronged approach.
\end{remark}

\begin{remark}[Information on right singular vectors]\label{remark-right-sing-vecs}
    Observe that the refined net term is critical to showing that the probability converges, and is \textit{only} magnified by an exponent to the appropriate order in the left singular vector case. Were one to repeat \cref{lem-lcd-fixed-val} for the right case, or use the two-pronged approach described in \cref{remark-one-sided} the refined net would not possess the proper exponent. 
    
    Additionally, now \textit{only} the probability term coming from the right singular vectors is aiding the proof, so we must ensure an $\LCD$ lower bound on the right at all times. To this end, we slice the right sub-level set. This is \textit{not} equivalent to the proof of the right event in Proposition 4.8 in \cite{christoffersen2025gaps}. That proof deals with the arrangement in which $\rLCD_{\gamma, \kappa}(u, \alpha) < \rLCD_{\gamma, \kappa}(v, \alpha)$ for a right-left pair $(u, v)$ with the $\rLCD$ of $v$ arbitrarily large, but \cref{lem-lcd-fixed-val} only establishes this fact for a polynomial upper bound on the $\rLCD$ of $v$, which is insufficient to conclude that any right singular vector has polynomially large $\rLCD$ with overwhelming probability.
\end{remark}

Now, we can establish polynomially large LCD lower bounds on right singular vectors.

\begin{proposition}\label{prop-lcd-varying-val}
    Let $M$ satisfy the assumptions of \cref{thm-repulsion-unscaled}, and let $\alpha$, $\kappa$ and $c$ satisfy the assumptions of \cref{lem-lcd-fixed-val}. Then, with probability at most $\exp(-O(p))$ we have that any left singular vector $w$ of $M$ satisfies $\rLCD_{\kappa, \gamma}(w, \alpha) \geq n^{c/\alpha}$.
\end{proposition}
\begin{proof}
    By \cref{prop-incompressible-varying-val} the probability that either a left or right singular vector are compressible is at most $\exp(-O(p))$.
    
    Suppose that $w$ is an incompressible left singular vector of $M$ associated with an incompressible right singular vector $v$ and singular value $\sigma$. Choosing a grid of spacing $\beta \sqrt{n}$ of the interval $[\sqrt{n} - K\sqrt{p}, \sqrt{n} + K\sqrt{p}]$, there exists a grid value $\sigma_0$ such that $|\sigma - \sigma_0| < \beta \sqrt{n}$ with probability at least $1 - \exp(-O(p))$ and so 
    \begin{equation*}
        \|M v - \sigma_0 u\| \leq \beta \sqrt{n}
    \end{equation*}
    By a union bound over $\sigma_0$ and an application of \cref{lem-lcd-fixed-val} we conclude the claim.
\end{proof}

\cref{prop-incompressible-varying-val} and \cref{prop-lcd-varying-val} result in the following Theorem, which rules out the main obstructions to large small-ball probabilities for the sake of eigenvalue repulsion.

\begin{theorem}\label{thm-obstructions}
     Let $M$ satisfy the assumptions of \cref{thm-repulsion-unscaled}, and let $\alpha$, $\kappa$ and $c$ satisfy the assumptions of \cref{lem-lcd-fixed-val}. Then with probability at least $1 - \exp(-O(p))$ we have
     \begin{enumerate}
         \item The left and right singular vectors incompressible (with $c_0$ and $c_1$ as in \cref{prop-incompressible-varying-val}).
         \item The left singular vectors admit an $\rLCD$ lower bound of $n^{c/\alpha}$.
     \end{enumerate}
\end{theorem}

Note that, as with \cite{christoffersen2025gaps}, the inclusion of the multiplicative term $1/\alpha$ in the exponent of the polynomial bound allows one to prove extremely strong lower bounds when $\alpha$ is taken to decay with $n$. However, we will only require that $\alpha$ is some small constant for our final repulsion result. Indeed, a slightly easier proof of \cref{lem-lcd-fixed-val} can be executed if one is willing to constraint $D$ by an upper bound of $n^{O(1)}$.

\subsection{Tall Wishart matrices have quantitative eigenvalue repulsion}
We now conclude our main repulsion results using the technique of \cite{nguyen2015randommatricestailbounds} and \cite{christoffersen2025gaps}.

\begin{theorem}[Squared singular value repulsion for tall real Wishart matrices]\label{theorem-repulsion-singular-vals-squared-real}
    Let $M$ satisfy the assumptions of \cref{thm-repulsion-unscaled}. There exists a $c \in (0, 1)$ such that, for $\alpha$ and $\kappa$ satisfying the assumptions of \cref{thm-obstructions}, and $\delta \geq n^{-c/\alpha}$, the event
    \begin{align*}
        \event_i &\coloneqq \{ \sigma_i^2(M) - \sigma_{i + 1}^2(M) \leq \sigma_{i + 1}(M) \cdot \delta \cdot p^{-1/2} \}
    \end{align*}
    satisfies
    \begin{equation}
         \sup_{1 \leq i \leq p}\mathbb{P}[\event_i] \leq O( \frac{\delta}{\sqrt{\alpha}})
    \end{equation}
    with implied constants depending only on $c_0$, $c_1$ and the law of $\xi$. Additionally, we have that the event
    \begin{align*}
        \mathcal{F}_i &\coloneqq \{ \sigma_i^2(M) - \sigma_{i + 1}^2(M) \leq (\sqrt{n} - K\sqrt{p}) \cdot \delta \cdot p^{-1/2} \}
    \end{align*}
    satisfies
    \begin{equation}
         \sup_{1 \leq i \leq p}\mathbb{P}[\mathcal{F}_i']  \leq \sup_{1 \leq i \leq p}\mathbb{P}[\mathcal{F}_i] \leq O( \frac{\delta}{\sqrt{\alpha}}) + \exp(-\Theta(p))
    \end{equation}
\end{theorem}
\begin{proof}
    Fix some $1 \leq i \leq p$. Decompose the matrix $M$ as $$M = (M'| X)$$ where $X$ is the $p$'th column, and the right singular vector corresponding to the $i$'th singular value of $M$ as 
    $$
    u = \binom{u'}{b}
    $$
    Additionally, define the events
    \begin{align*}
        \event_i' &\coloneqq \{ \sigma_i^2(M) - \sigma_{i + 1}^2(M') \leq \sigma_{i + 1}(M') \cdot \delta \cdot p^{-1/2} \} \\
        \mathcal{F}_i' &\coloneqq \{ \sigma_i^2(M) - \sigma_{i + 1}^2(M') \leq (\sqrt{n} - K\sqrt{p}) \cdot \delta \cdot p^{-1/2} \}
    \end{align*}
    The singular value equation can then be written as 
    \begin{equation*}
        \begin{pmatrix}
            M'^\top M' & M'^\top X \\
            X^\top M' & X^\top  X
        \end{pmatrix}
        \binom{u'}{b}
        = \sigma_i^2(M) \binom{u'}{b}
    \end{equation*}
    The top $p - 1$ rows of the above equation are equivalent to 
    \begin{equation*}
        M'^\top M'u' - \sigma_i^2(M)u'  = - M'^\top X b
    \end{equation*}
    Letting $\hat{u}$ be the right singular vector corresponding to the $i$'th singular value of $M'$, by multiplying on the left by the transpose of $\hat{u}$ we observe that
    \begin{equation*}
        (\sigma_i^2(M') - \sigma_i^2(M))\hat{u}^\top  u' = - b \sigma_i(M')\hat{w}^\top  X
    \end{equation*}
    where $\hat{w}$ is a left singular vector of $M'$ corresponding to $\sigma_i(M')$. Taking absolute values and applying Cauchy-Schwartz on the left, we derive the inequality
    \begin{equation*}
        |\hat{w}^\top  X| \leq \bigg| \frac{\sigma_i^2(M') - \sigma_i^2(M)}{b\sigma_i(M')} \bigg|
    \end{equation*}
    Thus, on the event $\event_i'$ we have that
    \begin{equation*}
         |\hat{w}^\top  X| \leq \bigg| \frac{\delta p^{-1/2}}{b} \bigg|
    \end{equation*}
    where $X$ is independent of the left singular vector $\hat{w}$ of $M'$, and $b$ is a coordinate of a right singular vector of $M$. Since $u \in \Incomp_p(c_0, c_1)$ it is well spread (in the sense of \cref{def-lcd}), so with probability at least $1 - \exp(-O(p))$ we have that at least $c_0c_1^2 p/2$ coordinates of $u$ have magnitude at least $c_1 p^{-1/2}/\sqrt{2}$. Following \cite{christoffersen2025gaps} we define the threshold value and events: 
    \begin{align*}
        B \coloneqq c_1 p^{-1/2}/\sqrt{2} \, ,
        \mathcal{G}_{ij} \coloneqq \event_i \wedge \{u_j \geq B \} \, ,
        p_B \coloneqq |\{j : u_j \geq B \}|
    \end{align*}
    where $p_B$ can be written as $\sum_{j = 1}^{p}l_{j}$ where $l_j$ is the indicator function of the event $\{u_j \geq B \}$. Decomposing the event $\event_i'$, by monotonicity and an application of Markov's inequality we have that for any $N \geq 1$ that 
    \begin{equation*}
        \mathbb{P}[\event_i'] \leq \frac{p}{N}\mathbb{P}[\mathcal{G}_{ip}] + \mathbb{P}[p_B < N]
    \end{equation*}
    since $\mathbb{P}[\mathcal{G}_{ij}] = \mathbb{P}[\mathcal{G}_{ip}]$ for all $1 \leq i, j \leq p$. Letting $N = c_0c_1^2 p/2$, we have that 
    \begin{align*}
        \mathbb{P}[\event_i'] \leq \frac{p}{N}\mathbb{P}[|\hat{w}^\top  X| \leq \frac{\delta p^{-1/2}}{B}] + \mathbb{P}[p_B < N] = \frac{2}{c_0 c_1^2}\mathbb{P}[|\hat{w}^\top  X| \leq \delta\frac{\sqrt{2}}{c_1}] + Ce^{-\Theta(p)}
    \end{align*}
    Now, by choosing a constant $K$ sufficiently large, which depends only on $c_0$ and $c_1$, we have that
    \begin{align*}
        \mathbb{P}[|\hat{w}^\top  X| \leq \delta\frac{\sqrt{2}}{c_1}] \leq Ce^{-\Theta(n)} + \LevyFunc(\hat{w}^\top  X, K \sqrt{2}\delta / c_1) \leq O(\frac{\delta}{\sqrt{\alpha}}) + C' \exp(-\Theta(n))
    \end{align*}
    since $\alpha$ has a constant upper bound. The analogous bound holds for $\event_i$ by the Cauchy interlacing theorem. The final observation for events $\mathcal{F}_i$ and $\mathcal{F}_i'$ then follows from the high probability lower bound on the singular values of $M$.
\end{proof}

\begin{corollary}[Extreme eigenvalue gap for tall real Wishart matrices]\label{cor-extreme-gap-real}
    For $M$ as in \cref{theorem-repulsion-singular-vals-squared-real}, we have that 
    \begin{equation*}
        \inf_{1 \leq i \leq p}(\lambda_i(M^\top  M) - \lambda_{i + 1}(M^\top  M)) \geq (\sqrt{n} - K\sqrt{p})p^{-3/2 - o(1/\log(n))}
    \end{equation*}
    with high probability. 
\end{corollary}

\subsection{The complex case}\label{subsec-repulsion-complex}
In this section, we briefly describe the modifications required to achieve the similar eigenvalue repulsion result for complex Wishart matrices. The general strategy will be to embed the complex objects $M \in \C^{n \times p}$ and $u \in C^{p}$ in $\R^{2n \times 2p}$ and $\R^{2p}$ respectively, in the vein of Section 6.3 of \cite{nguyen2015randommatricestailbounds}. It appears that the real representations of complex objects we use, as well as most of their properties, are folklore in the random matrix theory community. 

Indeed, it is often the case in random matrix theory that the complex analogues of real results are regarded as straightforward corollaries. However, in the Wishart matrix case, the newfound dependencies introduced in the real matrix presentations cause complications. The authors of \cite{christoffersen2025gaps} do not address the complex case, which we believe to be principally due to the fragile independence structure of their Lemma 4.4.

For an element $M \in \C^{n \times p}$ denote by $\Re(M), \Im(M) \in \R^{m \times n}$ the real and complex parts of $M$, and define the block-real representation of $M \in C^{m \times n}$ as the matrix 
\begin{equation*}
    \rho(M) \coloneqq         
    \begin{pmatrix}
        \Re(M) & -\Im(M) \\
        \Im(M) & \Re(M)
    \end{pmatrix}
\end{equation*}
For vectors $x \in \C^n$ we use the flat real representation 
\begin{equation*}
    r(x) \coloneqq (\Re(x), \Im(x))
\end{equation*}
Using the properties of these representations, we find that the approximate eigenvalue equation has the equivalence
\begin{equation}
    \| M^*M u - \sigma^2 u \| = \| \rho(M)^\top  \rho(M) r(u) - \sigma^2 r(u) \|
\end{equation}
and similarly for the right singular vector equation
\begin{equation}
    \| Mx - a \| = \|\rho(M)r(x) - r(a) \|
\end{equation}
While we are unable to directly apply the content from the real case due to the coordinate-wise dependence of $\rho(M)$, the dependence structure of $\rho(M)$ is rather simple, and minimal modifications can be made to the key Lemmas, which we describe now.

In \cref{lem-incompressible-fixed-val-right}, the decomposition $M = (X | Y)$ is not quite sufficient as $X$ and $Y$ possess some nontrivial dependency regardless of $c_0$. However, with $c_0 < 1/2$ we may further decompose $Y = (Y_1 | Y_2)$ where $Y_1$ is a $n \times (1/2 - c_0)$ matrix, so the concatenation of $X$ and $Y$ is equal to $(\Re(M) | -\Im(M))$ and thus $X$ and $Y_1$ are independent. We can then replace the lower bound $\|Y^\top  X u_0'\|$ with $\|Y_1^\top  X u_0'\|$, and the proof proceeds identically. 

An analogue of \cref{lem-incompressible-fixed-val-right} can be achieved in a similar way. Take the horizontal decomposition $M = \binom{X}{Y}$ and further decompose $Y = \binom{Y_1}{Y_2}$ such that $Y_1$ is a $(1/2 - c_0)n \times p$ matrix which independent of $X$, proceeding with the proof as before. Note that unlike the right singular vector setting above, in which $X$ and $Y_1$ have i.i.d.\ entries, the $X$ and $Y_1$ involved in the left singular vector case will only have i.i.d.\ sub-Gaussian rows, which is still sufficient to achieve the singular value bounds (see \cref{prop-sing-val-concentration}).

Lastly, \cref{lem-approx-singular} can be achieved with the new dependence structure by dropping the left prong of the argument of Lemma 4.4 of \cite{christoffersen2025gaps} and using a slightly different decomposition. Observe that the proof of Lemma 4.4 involves the decomposition 
\begin{equation*}
    M \coloneqq         
    \begin{pmatrix}
        A & B \\
        C & D
    \end{pmatrix}
\end{equation*}
where $A$ is an $(n - \lfloor \alpha n \rfloor) \times (p - \lfloor \alpha p \rfloor)$ matrix. Instead, we consider the same decomposition where $A$ is an $(n/2) \times (p - \lfloor \alpha p \rfloor)$ matrix, in which case the lower bound 
\begin{equation*}
    \|Mx - a \| \geq \|Ax' + Bx'' - a' \|
\end{equation*}
possesses $A$ and $B$ independent, and so the application of \cref{lem-tensorization} yields an upper bound 
\begin{equation*}
    \sup_{a \in \R^n}\mathbb{P}[ \{\|Mx - a\| < \varepsilon \sqrt{n}\}] \leq C\bigg(\frac{\varepsilon\sqrt{2}}{\gamma c_{1} \sqrt{\alpha}} + e^{-\Theta(\kappa^2)} \bigg)^{n/2}
\end{equation*}
which has no substantive impact on \cref{lem-lcd-fixed-val}. The dimensions of the anologous decomposition in \cite{christoffersen2025gaps} are crucial as it is necessary to achieve independence of the two events contributing the two probabilities which form the upper bound. We believe this to be the primary obstruction to \cite{christoffersen2025gaps} addressing the complex case, but since we only care about one half of the argument we need not maintain the same fragile decomposition structure.

Finally, we conclude the repulsion argument in a similar manner to \cref{theorem-repulsion-singular-vals-squared-real}. By an identical argument, we derive the inequality 
\begin{equation*}
    |\hat{w}^\top  X| \leq \bigg| \frac{\sigma_i^2(M') - \sigma_i^2(M)}{b\sigma_i(M')} \bigg|
\end{equation*}
where the $|\cdot|$ is the complex modulus. On the analogue of the event $\event_i$ this reduces to the inequality 
\begin{equation*}
    |\hat{w}^\top  X| \leq \frac{\delta p^{-1/2}}{|b|}
\end{equation*}
Now, following \cite[Remark 6.3]{nguyen2015randommatricestailbounds} and letting $w = (\Re(\hat{w}), \Im(\hat{w}))$, $w' = (\Im(\hat{w}), -\Re(\hat{w}))$ and $Y = (\Re(X), \Im(X))$ the above event reduces to 
\begin{equation*}
    \{ |w^\top  Y| \leq \frac{\delta p^{-1/2}}{|b|} \wedge |w'^\top  Y| \leq \frac{\delta p^{-1/2}}{|b|} \} \subseteq \{ |w^\top  Y| \leq \frac{\delta p^{-1/2}}{|b|} \}
\end{equation*}
Lastly, the term $b$ is a (complex) coordinate of $\hat{w}$, and so $|b|$ is the sum of the absolute value of two coordinates of an incompressible real vector $r(\hat{w}) = w$, and so the same argument from \cref{theorem-repulsion-singular-vals-squared-real} applies and we deduce that $\event_i$ occurs with probability at most $O(\delta/\sqrt{\alpha})$. Altogether, this yields the complex analogue of \cref{theorem-repulsion-singular-vals-squared-real}.

\appendix

\section{Background on the effective and efficient manipulation of algebraic numbers}\label{sec:algebraicNumbers}
Our algorithm for testing unitary tensor isomorphism exactly crucially depends on an ability to effectively and efficiently manipulate algebraic numbers (values $\alpha \in \AA$ that are roots of single variable polynomials with integer coefficients) at perfect precision. There is a well-developed literature concerning the perfect and efficient manipulation of algebraic numbers (see \cite{AGAlgorithms,cohen2013course}). We shall state a few basic and important facts about this theory.

First, associated with any algebraic number $\alpha \in \AA$ is a defining polynomial $p_\alpha$ that is the unique polynomial (coefficients do not have common factors) of least degree which vanishes at $\alpha$. Any algebraic number, even the irrational ones, then possess a standard finite encoding composed of its defining polynomial, and rational approximations of its real and imaginary parts of precision sufficient to distinguish $\alpha$ from the other roots of $p_\alpha$: a representation is a tuple $(p_\alpha, a, b, r) \in \mathbb{Z}[x] \times \mathbb{Q}^3$ where $r$ is the radius specifying the root contained around the point $a+bi$. A separation bound of Mignotte \cite{mignotte1982some} makes this representation rapidly computable:
for distinct roots $\alpha, \beta$ of polynomial $p \in \mathbb{Z}[x]$
\[
|\alpha - \beta| > \frac{\sqrt{6}}{d^{(d+1)/2}H^{d-1}},
\]
where $d$ is the degree and $H$ is the height of $p$. In particular, when $r$ is sufficiently small we have equality checking between algebraic numbers. Moreover, computing $\alpha+\beta$, $\alpha\beta$, $1/\alpha$, $\overline{\alpha}$, $|\alpha|$, $\sqrt{\alpha}$, as well as deciding whether $\alpha > \beta$ for distinct algebraic $\alpha$ and $\beta$, may all be performed in polynomial-time (in representation lengths).

\section{Proof sketch of Corollary~\ref{cor:hypergraph-iso}}
\label{app:hypergraph-iso}

Let $L=M=N=[n]:=\{1, 2, \dots, n\}$, and let $G, H\subseteq L\times M\times N$ be the edge sets of two 3-uniform tripartite hypergraphs on the vertex set $L\cup M\cup N$. Let $A_G\in \RR^n\otimes\RR^n\otimes\RR^n$ be the adjacency tensor of $G$, such that $A_G(i, j, k)=1$ if $(i, j, k)\in G$, and $-1$ otherwise. Then testing whether $G$ and $H$ are isomorphic is equivalent to testing whether $A_G$ and $A_H$ are isomorphic under the natural action of $\Sym_n\times\Sym_n\times\Sym_n$ on $\RR^n\otimes\RR^n\otimes\RR^n$, where $\Sym_n$ is the symmetric group over $[n]$.

Suppose $G$ is a random hypergraph, that is, each edge $(i, j, k)$ is included in $G$ with probability $1/2$. Let $F_G$ be the flattening of $A_G$ along the vertext set $N$ direction, so $F_G\in \M(n^2\times n, \RR)$. By Theorem~\ref{thm-repulsion-unscaled}, with high probability over $F_G$ whose entries are drawn from the Rademacher distribution, eigenvalues of $F_G^\top F_G$ are of multiplicity $1$, and they are well-separated. Let $F_H$ be the flattening of $A_H$ along the vertex set $N$ direction. For $G$ and $H$ to be isomorphic, a necessary condition is that the singular values of $F_G$ and those of $F_H$ are the same. 

Let $C_N=\{\pi\in \Sym_n\mid \pi^\top F_G^\top F_G \pi=F_H^\top F_H\}$. Note that $C_N$ is a coset of the group $B_N:=\{\pi\in \Sym_n\mid \pi^\top F_G^\top F_G \pi=F_G^\top F_G\}$. By \cite{leighton1979certificates} (see also \cite{spielman2018testing}), we can compute $C_N$ represented by a coset representative $\sigma\in \Sym_n$ and a generating set of $B_N$. Furthermore, $B_N$ is a $2$-group, that is, $|B_N|$ is a power of $2$. 

We then carry out the above process along the $M$ and $L$ directions to get $C_M$ and $B_M$, as well as $C_L$ and $B_L$. After applying the coset representatives to transform $A_G$ to $A_G'$, the problem becomes to test if there exists an element from $B_L\times B_M\times B_N$ that sends $A_G'$ to $A_H$. Note that $B_L\times B_M\times B_N$ is a $2$-group. This problem, the set-wise transporter problem under $2$-groups, can then be solved by Luks in his celebrated result of bounded-valence graph isomorphism \cite{Luks82}.

\paragraph{Acknowledgement.} The authors would like to thank Peter B\"urgisser for his interest and discussions on this paper.
Y.Q. would like to thank Robert Andrews for sharing insights into Gr\"obner basis computation over $\RR$. S.E. was supported by the National Science Foundation Graduate Research Fellowship Program under Grant No. 2140001.

\bibliographystyle{alphaurl}
\bibliography{refs}

\end{document}